\documentclass[dvipdfmx]{article}
\usepackage{amsmath,amsthm,amssymb}
\usepackage{mathrsfs}

\newtheorem{theo}{Theorem}[section]
\newtheorem{prop}[theo]{Proposition}
\newtheorem{coro}[theo]{Corollary}
\newtheorem{lemm}[theo]{Lemma}

\theoremstyle{definition}
\newtheorem{defi}[theo]{Definition}
\newtheorem*{note*}{Note}
\newtheorem*{claim*}{Claim}
\newtheorem*{exam*}{Example}
\newtheorem*{rema*}{Remark}

\author{Shota FUKUSHIMA\thanks{Graduate School of Mathematical Sciences, the University of Tokyo, 3-8-1 Komaba, Meguro-ku, Tokyo, 153-8914, Japan. 
Email: fukusima@ms.u-tokyo.ac.jp} 
\thanks{The author is supported by Leading Graduate Course for Frontiers of Mathematical Sciences and Physics (FMSP), at Graduate School of Mathematical Science, the University of Tokyo. }}
\title{Time-slicing approximation of Feynman path integrals on compact manifolds}
\date{\empty}

\newcommand{\transp}[1]{{}^t\!{#1}}

\newcommand{\vol}{\mathrm{vol}}
\newcommand{\jbracket}[1]{\left\langle {#1} \right\rangle}

\newcommand{\rmop}[1]{\mathop{\mathrm{#1}}}

\numberwithin{equation}{section}

\begin{document}
\maketitle
\begin{abstract}
We construct fundamental solutions to the time-dependent Schr\"odinger equations on compact manifolds by the time-slicing approximation of the Feynman path integral. We show that the iteration of short-time approximate solutions converges to the fundamental solutions to the Schr\"odinger equations modified by the scalar curvature in the uniform operator topology from the  Sobolev space to the space of square integrable functions. In order to construct the time-slicing approximation by our method, we only need to consider broken paths consisting of sufficiently short classical paths. We prove the convergence to fundamental solutions by proving two important properties of the short-time approximate solution, the stability and the consistency.  
\end{abstract}


\section{Introduction}
\subsection{Main theorem}
Let $(M, g)$ be a compact connected oriented smooth Riemannian manifold and consider the initial value problem of the semiclassical Schr\"odinger equation
\[
i\hbar\frac{\partial}{\partial t}u(t, x)=-\frac{\hbar^2}{2}\triangle_gu(t, x)+V(x)u(t, x), \quad u(0, x)=u(x), \quad (t, x)\in \mathbb{R}\times M 
\]
with a smooth potential $V: M\to \mathbb{R}$ and $\hbar\in (0, 1]$. 
We want to construct a fundamental solution to this problem by the time-slicing approximation of the Feynman path integrals \cite{Feynman48}. It is known that it is impossible to construct a complex measure which realizes Feynman path integrals for Schr\"odinger equations \cite{Cameron60}. However, an alternative method, the time-slicing approximation method, which is the original idea of Feynman \cite{Feynman48}, still has the possibility to be justified mathematically. On the Euclidean space, the time-slicing approximation is studied by Fujiwara \cite{Fujiwara79} \cite{Fujiwara17}, Kumano-go \cite{Kumano-go04} and Ichinose \cite{Ichinose20}. On the other hand, there are only a few studies of mathematical analysis of the Feynman path integrals on curved spaces (\cite{Fujiwara76}, \cite{Miyanishi15}). In physics, it is pointed by DeWitt \cite{DeWitt57} that the natural Feynman quantization on curved spaces gives the solution to a modified Schr\"odinger equation
\[
i\hbar\frac{\partial}{\partial t}u(t, x)=-\frac{\hbar^2}{2}\triangle_gu(t, x)+V(x)u(t, x)+\frac{\hbar^2}{12}R(x)u(t, x), 
\]
where $R: M\to \mathbb{R}$ is the scalar curvature of $(M, g)$. Our result is a generalization of \cite{Miyanishi15} to general compact manifolds and general smooth potentials. Time-slicing approximations of heat kernels  for supersymmetric quantum mechanical systems on manifolds is discussed in \cite{Fine-Sawin17}. 

In the following, we introduce our mathematical setting of the time-slicing approximation. Let 
\[
H(x, \xi)=\frac{1}{2}|\xi|_g^2+V(x), \quad (x, \xi)\in T^*M. 
\]
be the corresponding classical Hamiltonian. 
We assume that $V: M\to \mathbb{R}$ is smooth. For $(x, y)\in M\times M$ and $t>0$, we denote by $\Gamma^H_{t, x, y}$ the set of all the classical paths $(x(s), \xi(s)): [0, t]\to T^*M$ with respect to the Hamiltonian $H$ such that $x(0)=y$ and $x(t)=x$. In general the set $\Gamma^H_{t, x, y}$ does not consist of a single element. The simplest counterexample is $M=S^1$ and $V=0$. 
However, if we take a small neighborhood $N$ of the diagonal $\mathrm{diag}\,(M):=\{(x, x) \in M\times M \mid x\in M\}$ and small time $t>0$, then for an arbitrary pair $(x, y)\in N$, there exists a unique classical path $(x(s), \xi(s))\in \Gamma^H_{t, x, y}$ such that  
\[
|\xi(0)|_g=\min_{(q(s), p(s))\in \Gamma^H_{t, x, y}} |p(0)|_g. 
\]
(We prove this fact in Section \ref{sect_low_energy_cm}.)
We define $(x^t_s(x, y), \xi^t_s(x, y))\in \Gamma^H_{t, x, y}$ as the unique minimizer of the initial momentum $|\xi(0)|_g$. Then we can define the action function 
\[
S(t, x, y):=\int_0^t \left(\frac{1}{2}\left|\frac{d}{ds}x^t_s(x, y)\right|_g^2-V(x^t_s(x, y))\right) \,ds
\]
along the lowest energy path. We remark that ``\textit{the lowest energy}'' is equivalent to  ``\textit{the lowest momentum}'' because we fix an initial point $y\in M$. 

Next we introduce the Morette-Van Vleck determinant \cite{Morette51}. Take local coordinates $\varphi_\lambda=(x_1, \ldots, x_n)$ near $x$ and $\varphi_\mu=(y_1, \ldots, y_n)$ near $y$. Then we define positive definite matrices $(g^\lambda_{jk}(x))_{j, k=1}^n$ and $(g^\mu_{jk}(y))_{j, k=1}^n$ by 
\[
g_x=\sum_{j, k=1}^n g^\lambda_{jk}(x)dx_j dx_k, \quad g_y=\sum_{j, k=1}^n g^\mu_{jk}(y)dy_j dy_k
\]
respectively. 
The Morette-Van Vleck determinant is defined as 
\begin{equation}\label{eq_morette_van_vleck}
D(t, x, y):=g_\lambda(x)^{-\frac{1}{2}}g_\mu(y)^{-\frac{1}{2}}\det \left(-\frac{\partial^2 S}{\partial x_j\partial y_k}(t, x, y)\right)_{j, k=1}^n, 
\end{equation}
where $g_\lambda(x)=\det (g^\lambda_{jk}(x))_{j, k=1}^n$ and $g_\mu(y)=\det (g^\mu_{jk}(y))_{j, k=1}^n$. 
This is independent of the choice of local coordinates. If we take the neighborhood $N$ of $\rmop{diag}(M)$ sufficiently small, then $D(t, x, y)$ is positive for small $t>0$ and $(x, y)\in N$ (a consequence of Proposition \ref{prop_amplitude_diagonal}). 

Now we consider the modified semiclassical Schr\"odinger operator 
\begin{equation}
\tilde H_\hbar=-\frac{\hbar^2}{2}\triangle_g+V(x)+\frac{\hbar^2}{12}R(x) \label{eq_modified_hamiltonian}
\end{equation}
and the corresponding Schr\"odinger equation
\begin{equation}
i\hbar\frac{\partial}{\partial t}u(t)=\tilde H_\hbar u(t), \quad u(0)=u_0. \label{eq_schrodinger}
\end{equation}
Since $M$ is compact and $V$ is smooth, $\tilde H_\hbar$ is essentially self-adjoint on $L^2(M)$. We also denote its unique self-adjoint extension by $\tilde H_\hbar$. By Stone's theorem, there exists a family of unitary operators $\{ e^{-it\tilde H_\hbar/\hbar}\}_{t\in\mathbb{R}}$ on $L^2(M)$ such that $u=e^{-it\tilde H_\hbar/\hbar}u_0\in C^1(\mathbb{R}; L^2(M))$ gives the unique solution to the initial value problem \eqref{eq_schrodinger} for $u_0\in H^2(M)$ (the Sobolev space of degree 2). We define an approximate solution to the Schr\"odinger equation as 
\[
E_\hbar^\chi(t)u(x):=\frac{1}{(2\pi i\hbar)^{n/2}}\int_M \chi(x, y) \sqrt{D(t, x, y)} e^{iS(t, x, y)/\hbar} u(y)\, \vol_g(y). 
\]
Here $\chi\in C_c^\infty(N; [0, 1])$ with $\chi=1$ near $\mathrm{diag}(M)$ and $\vol_g$ is the natural volume form associated with the metric $g$. If $\chi$ is obvious from context, we drop $\chi$ and denote $E_\hbar^\chi(t)=E_\hbar (t)$. 

A partition $\Delta$ of $t>0$ is a finite sequence $\Delta=(t_1, \ldots, t_L)$ such that $t_j>0$ for all $j$ and $t_1+\cdots+t_L=t$. The size of partition is defined as $|\Delta|:=\max_{1\leq j \leq L}t_j$. We define the time-slicing approximation of the Feynman path integral as
\[
\mathcal{E}_\hbar^\chi(\Delta)=\mathcal{E}_\hbar(\Delta):=E_\hbar(t_L)\cdots E_\hbar(t_1). 
\]
$\mathcal{E}_\hbar(\Delta)$ is interpreted as the sum or an approximate Feynman path integral over broken paths consisting of short classical segments. 

The semiclassical Sobolev space $H_\hbar^s (M)$ of degree $s\in \mathbb{R}^n$ is a Hilbert space defined as 
\[ H_\hbar^s(M):=\{\, u\in \mathscr{D}^\prime (M) \mid (-\hbar^2 \triangle_g +1)^{s/2}u\in L^2(M)\,\}\]
with an inner product 
\[
    \jbracket{u, v}_{H_\hbar^s}:=\jbracket{ (-\hbar^2 \triangle_g+1)^{s/2}u, (-\hbar^2 \triangle_g+1)^{s/2}v}_{L^2}. \]
Note that $H_\hbar^s(M)=H_1^s(M)$ as sets for all $\hbar\in (0, 1]$. The Sobolev norm is defined as
\begin{equation}\label{eq_sobolev_norm} \|u\|_{H^s_\hbar}^2:=\jbracket{ (-\hbar^2 \triangle_g+1)^{s/2}u, (-\hbar^2 \triangle_g+1)^{s/2}u}_{L^2}=\| (-\hbar^2 \triangle_g+1)^{s/2}u\|_{L^2}^2. \end{equation}

Our main theorem states that its iteration converges to the fundamental solution to the Schr\"odinger equation. 

\begin{theo}\label{theo_main}
There exists a neighborhood $N$ of $\rmop{diag}(M)$ such that the following statement holds. For all $T>0$, $\varepsilon\in (0, 1/2]$ and $\chi\in C_c^\infty(N; [0, 1])$ with $\chi=1$ near the diagonal, there exist positive constants $\delta>0$ and $C>0$ such that, for all $t\in (0, T]$, $\hbar\in (0, 1]$ and all partitions $\Delta$ of $t$ with $|\Delta|<\delta$, the inequality
\[
\left\| \left( \mathcal{E}_\hbar(\Delta)-e^{-it\tilde H_\hbar /\hbar} \right) (-\hbar^2\triangle_g+1)^{-(1+\varepsilon)/2}\right\|_{L^2\to L^2} \leq Ct|\Delta|^\varepsilon
\]
holds. In particular, as $|\Delta|\to 0$, $\mathcal{E}_\hbar (\Delta)$ converges to $e^{-it\tilde H_\hbar/\hbar}$ in the operator norm topology from the semiclassical Sobolev space $H_\hbar^{1+\varepsilon}(M)$ to $L^2(M)$. 
\end{theo}

The convergence in the strong operator topology on $L^2$ spaces is an immediately consequence from the Theorem \ref{theo_main} and the stability of the short-time approximate solution $E_\hbar (t)$, which is stated in Theorem \ref{theo_stability_consistency}. 

\begin{coro}\label{coro_strong_convergence}
Take a neighborhood $N$ of $\rmop{diag}M$ and $\chi\in C_c^\infty(N; [0, 1])$ as in Theorem \ref{theo_main}. Then for any $T>0$, $\mathcal{E}_\hbar(\Delta)$ converges to $e^{-it\tilde H_\hbar /\hbar}$ uniformly in $\hbar \in (0, 1]$ and $t\in (0, T]$ in the strong operator topology as $|\Delta|\to 0$: 
\[
\lim_{|\Delta|\to 0} \|\mathcal{E}_\hbar(\Delta)u-e^{-it\tilde H_\hbar/\hbar}u\|_{L^2}=0, \quad \forall u\in L^2(M). 
\]
More precisely, for any $T>0$, $u\in L^2(M)$ and $\varepsilon>0$, there exists $\delta^\prime>0$ such that, for any $t\in (0, T]$, partition $\Delta$ of $t$ with $|\Delta|<\delta^\prime$ and $\hbar \in (0, 1]$, the inequality 
\[\|\mathcal{E}_\hbar(\Delta)u-e^{-it\tilde H_\hbar/\hbar}u\|_{L^2}<\varepsilon
    \]
holds. 
\end{coro}

We emphasize two observations concerning this paper. First, we do not need to connect arbitrary two points in the configuration space in order to construct a short-time approximate solution $E_\hbar (t)$. Fujiwara \cite{Fujiwara79} \cite{Fujiwara17} used classical paths for constructing the Feynman path integrable on Euclidean spaces, proving that for arbitrary two points and short time, there exists a unique classical path connecting them if the potential $V$ is at most quadratically increasing. Kumano-go \cite{Kumano-go04} and Ichinose \cite{Ichinose20} used straight lines on Euclidean spaces. On the other hand, on manifolds, there may not exist any natural choices of paths connecting arbitrary two points even if the potential $V$ is zero. In this paper, by introducing a cutoff function whose support is included in some small neighborhood near the diagonal, we avoid this problem. This idea is already used in \cite{Miyanishi15}. 

Next, the topology of the convergence to fundamental solutions to Schr\"odinger equations is stronger than the previous results. For example, Miyanishi \cite{Miyanishi15} proved the convergence of the time-slicing approximation to the fundamental solution to Schr\"odinger equations for free particles ($V=0$) on compact rank 1 locally symmetric spaces. More precisely, he proved that 
\[
\rmop{s\text{-}\lim}_{N\to \infty} E_1 \left( \frac{t}{N} \right)^N P(N) = e^{-it\tilde H}, 
\]
where $P(\lambda)$ is the spectral projector of the self-adjoint operator $\tilde H=\tilde H_1$ defined by the spectral decomposition
\[
\tilde H=\int_{-\infty}^\infty \lambda\, dP(\lambda). 
\]
Ichinose \cite{Ichinose20} proved the convergence to fundamental solutions to Schr\"odinger equations corresponding to the Hamiltonian with polynomially increasing (and time-dependent) scalar and vector potentials on Euclidean spaces in the $L^2$ strong operator topology. 

We cannot expect stronger convergence under our setting, for instance the operator norm topology on $L^2$ spaces, due to the cutoff function in the definition of the short-time approximate solution $E_\hbar (t)$. Miyanishi \cite{Miyanishi14} remarks the non-uniformness of the convergence to a fundamental solution of the free Schr\"odinger equation on the sphere. The strong convergence of (imaginary-time) Feynman path integrals is also treated in \cite{Ichinose12}. 

Concerning the scalar curvature term $R/12$, we remark that Schulman \cite{Schulman81} mentions it in the physics point of views. A mathematical treatment of the curvature term is also discussed in \cite{Woodhouse92} in the context of geometric quantization.

\subsection{Outline of the proof}
In order to prove our main theorem following Ichinose \cite{Ichinose20}, we show that the approximate solution $E_\hbar (t)$ satisfies the following two properties\footnote{The terminology ``\textit{stability}'' and ``\textit{consistency}'' is also taken from \cite{Ichinose20}. }: 
\begin{itemize}
\item \textit{Stability}: $\|E_\hbar (t)\|_{L^2\to L^2}\leq 1+O(t)$, 
\item \textit{Consistency}: $i\hbar\partial_t E_\hbar(t)-\tilde H_\hbar E_\hbar(t)\to 0$ as $t\to +0$ in some sense. 
\end{itemize}
This procedure is similar to the Chernoff approximation, which is applied to the construction of the solution to the heat equations \cite{Chernoff68}. 

The (square root of) Morette-Van Vleck determinant $\sqrt{D(t, x, y)}$ defined by \eqref{eq_morette_van_vleck} is necessary to establish the stability (see Morette \cite{Morette51}). However, due to this amplitude, $E_\hbar (t)$ is not an approximate solution to the original Schr\"odinger equation: 
\[
\left(i\hbar\frac{\partial}{\partial t}-\left(-\frac{\hbar^2}{2}\triangle_g +V(x)\right)\right)(\sqrt{D}e^{iS})
\longrightarrow \frac{\hbar^2}{12}R(x)\delta_x(y) \quad (t\to +0) 
\]
in the distributional sense. Here $\delta_x(y)\in \mathscr{D}^\prime(M)$ is the delta function in the following sense: 
\[
\int_M \delta_x(y)u(y)\, \vol_g(y)=u(x). 
\]
Thus we need the modification $\tilde H_\hbar=-\hbar^2\triangle_g/2+V(x)+\hbar^2 R(x)/12$ in order to establish the consistency (see \cite{DeWitt57}). Another difficulty to derive the consistency is the existence of cutoff function $\chi(x, y)$. Due to $\chi$, we must ignore the propagation of the high energy component. The operator $(-\hbar^2\triangle_g+1)^{-(1+\varepsilon)/2}$ is inserted in order to dump the high energy component. 

The procedure of the rigorous proof is as follows. In Section \ref{sect_s_and_c}, we state the stability and consistency, and give the proof of Theorem \ref{theo_main} from them. In Section \ref{sect_low_energy_cm}, we investigate the classical mechanics on compact manifolds. The neighborhood $N$ of the diagonal in the statement of Theorem \ref{theo_main} is defined there. Section \ref{sect_FIO} is devoted to prepare general theories of oscillatory integral operators in order to define a short-time approximate solution $E_\hbar (t)$. In Section \ref{sect_def_of_E_h(t)}, we give a precise definition of the short-time approximate solution $E_\hbar (t)$ and prove the fundamental properties of it. Finally we prove the stability and the consistency in Section \ref{sect_proof_s_and_c}. 


\section{Reduction to stability and consistency}\label{sect_s_and_c}

In this section, we show Theorem \ref{theo_main} from the stability and the consistency. First we mention fundamental properties of the approximate solution $E_\hbar (t)$. The precise definition of $N\supset \rmop{diag}(M)$ and $t_0>0$ in the following proposition are presented in the statement of Theorem \ref{theo_action_well_defined}. 

\begin{prop}\label{prop_fundamental_properties}
Suppose that $\chi\in C_c^\infty(N)$ satisfies $\chi=1$ near the diagonal. Then $E_\hbar (t)=E^\chi_\hbar (t)$ satisfies the following properties. 
\begin{enumerate}
\renewcommand{\labelenumi}{(\roman{enumi})}
\item For all $u\in C^\infty(M)$ and $\hbar\in (0, 1]$, $t\in (0, t_0] \mapsto E_\hbar (t)u \in L^2(M)$ is continuous in $L^2(M)$. 
\item For all $u \in C^\infty(M)$ and $\hbar\in (0, 1]$, $\| E_\hbar (t)u-u\|_{L^2(M)}\to 0$ as $t\to +0$. 
\item For all $u\in C^\infty(M)$ and $\hbar\in (0, 1]$, $t\in (0, t_0] \mapsto E_\hbar (t)u\in L^2(M)$ is differentiable in $L^2(M)$ and $t\in (0, t_0]\mapsto \partial_t E_\hbar (t)u\in L^2(M)$ is continuous in $L^2(M)$. 
\end{enumerate}
\end{prop}

\begin{rema*}
    We does not state that the convergence in Proposition \ref{prop_fundamental_properties} (ii) is uniform in $\hbar\in (0, 1]$. We are only interested in the continuity of $t\mapsto E_\hbar (t)u$ in $L^2(M)$ at $t=0$. 
\end{rema*}

We prove Proposition \ref{prop_fundamental_properties} in Subsection \ref{subsect_proof_of_fp}. 
We define $E_\hbar (0)u=u$ for all $u\in L^2(M)$. Then (i) and (ii) imply that $t\in [0, t_0]\mapsto E_\hbar (t)u\in L^2(M)$ is continuous in $L^2(M)$ for all $u\in C^\infty(M)$ and $\hbar\in (0, 1]$. 

Now we state the stability and consistency. 

\begin{theo}\label{theo_stability_consistency}
There exists $\delta>0$ such that the following statements hold. 
\begin{enumerate}
\item (Stability) There exists $C>0$ such that for all $t\in (0, \delta]$ and any $\hbar\in (0, 1]$, the  inequality
\[
\|E_\hbar(t)\|_{L^2(M)\to L^2(M)}\leq e^{C\hbar t}
\]
holds. 
\item (Consistency) For all $\varepsilon\in (0, 1/2]$, there exists $C_\varepsilon>0$ such that for all $t\in (0, \delta]$, $\hbar\in (0, 1]$ all $u\in C^\infty(M)$, the inequality
\[
\left\| i\hbar\frac{\partial}{\partial t}E_\hbar (t)u-\tilde H_\hbar E_\hbar(t)u\right\|_{L^2(M)}\leq C_\varepsilon \hbar t^\varepsilon\|u\|_{H^{1+\varepsilon}_\hbar (M)}
\]
holds. 
\end{enumerate}
\end{theo}

We prove Theorem \ref{theo_stability_consistency} in Section \ref{sect_proof_s_and_c}. We now reduce the proof of Theorem \ref{theo_main} to that of Theorem \ref{theo_stability_consistency}. 

\begin{proof}[Proof of Theorem \ref{theo_main}]
Let $u\in C^\infty(M)$. We set 
\[
G_\hbar (t)u:=
\begin{cases}
i\hbar \partial_tE_\hbar (t)u-\tilde H_\hbar E_\hbar (t)u & \text{for } t>0, \\
0 & \text{for } t=0. 
\end{cases}
\]
The continuity of $t\in [0, t_0]\mapsto G_\hbar (t)u\in L^2(M)$ at $t=0$ is the consequence of the consistency. By the Duhamel principle and $E_\hbar (0)u=u$, we obtain
\[
E_\hbar (t)u-e^{-it\tilde H_\hbar/\hbar}u=-\frac{i}{\hbar}\int_0^t e^{-i(t-s)\tilde H_\hbar/\hbar}G_\hbar (s)u\, ds. 
\]
The consistency gives the estimate
\[
\|E_\hbar(t)u-e^{-it\tilde H_\hbar/\hbar}u\|_{L^2}
\leq \frac{1}{\hbar}\int_0^t \|G_\hbar(s)u\|_{L^2}\, ds
\leq C_\varepsilon t^{1+\varepsilon} \|u\|_{H_\hbar^{1+\varepsilon}}. 
\]
Since it holds for all $u\in C^\infty(M)$, we obtain
\[
\|E_\hbar(t)-e^{-it\tilde H_\hbar/\hbar}\|_{H_\hbar^{1+\varepsilon}\to L^2}
\leq C_\varepsilon t^{1+\varepsilon}. 
\]

Take $\mu\gg 1$ such that $V(x)+R(x)/12+\mu>0$ for all $x\in M$. Then the $L^2$-bounded operator $P_\hbar:=(\tilde H_\hbar+\mu)^{-(1+\varepsilon)/2}$ is well-defined and $P_\hbar L^2(M)=H_\hbar^{1+\varepsilon}(M)$. Thus, for all $u\in L^2(M)$, we have 
\[
\|(E_\hbar(t)-e^{-it\tilde H_\hbar/\hbar})P_\hbar u\|_{L^2}
\leq C_\varepsilon t^{1+\varepsilon}\|P_\hbar u\|_{H_\hbar^{1+\varepsilon}}
\leq C_\varepsilon t^{1+\varepsilon}\|u\|_{L^2}. 
\]
Hence
\[
\|(E_\hbar(t)-e^{-it\tilde H_\hbar\hbar})P_\hbar\|_{L^2\to L^2}
\leq C_\varepsilon t^{1+\varepsilon}. 
\]

For a partition $\Delta=(t_1, \ldots, t_L)$ of $t$ with $|\Delta|<\delta$, we have
\begin{align*}
&E_\hbar(\Delta)-e^{-it\tilde H_\hbar/\hbar} \\
&= 
\sum_{j=1}^L E_\hbar(t_L)\cdots E_\hbar(t_{j+1}) (E_\hbar(t_j)-e^{-it_j\tilde H_\hbar/\hbar}) e^{-i(t_{j-1}+\cdots +t_1)\tilde H_\hbar/\hbar}. 
\end{align*}
Here we interpret the summand as 
\begin{align*}
&E_\hbar(t_L)\cdots E_\hbar(t_{j+1}) (E_\hbar(t_j)-e^{-it_j\tilde H_\hbar/\hbar}) e^{-i(t_{j-1}+\cdots +t_1)\tilde H_\hbar/\hbar} \\
&=E_\hbar(t_L)\cdots E_\hbar(t_2)(E_\hbar(t_1)-e^{-it_1\tilde H_\hbar/\hbar}) 
\end{align*}
if $j=1$ and 
\begin{align*}
&E_\hbar(t_L)\cdots E_\hbar(t_{j+1}) (E_\hbar(t_j)-e^{-it_j\tilde H_\hbar/\hbar}) e^{-i(t_{j-1}+\cdots +t_1)\tilde H_\hbar/\hbar} \\
&=(E_\hbar(t_L)-e^{-it_N\tilde H_\hbar/\hbar})e^{-i(t_{N-1}+\cdots+t_1)\tilde H_\hbar/\hbar}
\end{align*}
if $j=L$. 
We have 
\begin{align*}
&\|E_\hbar(t_L)\cdots E_\hbar(t_{j+1}) (E_\hbar(t_j)-e^{-it_j\tilde H_\hbar/\hbar}) e^{-i(t_{j-1}+\cdots +t_1)\tilde H_\hbar/\hbar}P_\hbar\|_{L^2\to L^2} \\
&\leq e^{C(t_{j+1}+\cdots+t_L)}\|(E(t_j)-e^{-it_j\tilde H_\hbar/\hbar})P_\hbar\|_{L^2\to L^2} 
\leq C_\varepsilon t_j^{1+\varepsilon}
\end{align*}
by the consistency, stability and the fact that $e^{-it_j\tilde H_\hbar/\hbar}$ and $P_\hbar$ commute. 
Summing up these inequalities, we obtain the estimate
\[
\|(E_\hbar(\Delta)-e^{-it\tilde H_\hbar/\hbar})P_\hbar\|_{L^2\to L^2} 
\leq \sum_{j=1}^L C_\varepsilon t_j^{1+\varepsilon}
\leq C_\varepsilon \sum_{j=1}^L t_j |\Delta|^\varepsilon
\leq C_\varepsilon t|\Delta|^\varepsilon. \qedhere
\]
\end{proof}

\begin{proof}[Proof of Corollary \ref{coro_strong_convergence}]
Let $u\in L^2(M)$ and $\varepsilon^\prime>0$. Since $C^\infty(M)$ is dense in $L^2(M)$, we can take a function $v \in C^\infty(M)$ such that $\|u-v\|_{L^2}<\varepsilon^\prime$. By Theorem \ref{theo_main} and the stability of $E_\hbar (t_j)$, we have
\begin{align*}
&\|\mathcal{E}_\hbar (\Delta)u-e^{-it\tilde H_\hbar/\hbar}u\|_{L^2} \\
&\leq \|\mathcal{E}_\hbar (\Delta)v-e^{-it\tilde H_\hbar/\hbar}v\|_{L^2}
+(\|\mathcal{E}_\hbar (\Delta)\|_{L^2\to L^2} 
+\| e^{-it\tilde H_\hbar/\hbar}\|_{L^2\to L^2}) \|u-v\|_{L^2} \\
&\leq Ct|\Delta|^{1/2} \|v\|_{H^{3/2}_\hbar}+\left(\prod_{j=1}^L e^{C\hbar t_j}+1\right)\varepsilon^\prime \\
&\leq Ct|\Delta|^{1/2} \| v\|_{H^{3/2}_\hbar}+(e^{CT}+1)\varepsilon^\prime. 
\end{align*}
We take $\delta^\prime>0$ such that $CT(\delta^\prime)^{1/2} \| v\|_{H^{3/2}_1}<\varepsilon^\prime$. Note that $\|w\|_{H^{3/2}_\hbar}\leq \| w\|_{H^{3/2}_1}$ for all $w\in H_\hbar^s(M)$ and $s\geq 0$ since the spectral decomposition 
\[
    -\triangle_g=\int_0^\infty \lambda \, dP^\prime (\lambda)\]
and the definition of the Sobolev norm \eqref{eq_sobolev_norm} imply 
\[
    \|w\|_{H^s_\hbar}^2 \leq \int_0^\infty (1+\hbar^2 \lambda)^s\, d\jbracket{P^\prime (\lambda)w, w}\leq \int_0^\infty (1+\lambda)^s\, d\jbracket{P^\prime (\lambda)w, w}=\|w\|_{H^s_1}^2. \]
If $|\Delta|<\delta^\prime$, $\hbar\in (0 ,1]$ and $t\in (0, T]$, then 
\[
\|\mathcal{E}_\hbar (\Delta)u-e^{-it\tilde H_\hbar/\hbar}u\|_{L^2}
\leq (e^{CT}+2)\varepsilon^\prime. 
\]
Therefore $\mathcal{E}_\hbar (\Delta)u\in L^2(M)$ converges to $e^{-it\tilde H_\hbar /\hbar}u\in L^2(M)$ uniformly in $\hbar \in (0, 1]$ and $t\in (0, T]$. 
\end{proof}


\section{Low energy classical mechanics}\label{sect_low_energy_cm}

We prove the existence and uniqueness of the classical path connecting sufficiently close two points in sufficiently short time. The action integral along the classical path is defined and satisfies suitable conditions for phase functions (Theorem \ref{theo_phase_function}). We also investigate the properties of the Morette-Van Vleck determinant in Subsection \ref{subsec_mvv}.

\subsection{Existence and uniqueness of the lowest energy classical path}\label{subsec_classical_mechanics}

We first introduce sets of classical paths connecting two points with low energy. 

\begin{defi}
For a Hamiltonian $H: T^*M\to \mathbb{R}$, $(t, x, y)\in (0,\infty)\times M\times M$ and $\mu\in\mathbb{R}$, we define 
\[
\Gamma^H_{t, x, y}(\mu):=\{\, (x(s), \xi(s))\in \Gamma^H_{t, x, y} \mid \, |\xi(0)|_g<\mu\,\}. 
\]
Here $\Gamma^H_{t, x, y}$ is the set of all the classical paths $(x(s), \xi(s)): [0, t]\to T^*M$ with respect to the Hamiltonian $H$ such that $x(0)=y$ and $x(t)=x$.

If $\Gamma^H_{t, x, y}(\mu)$ consists of a single element, then we call the element the lowest energy classical path from $y$ to $x$ in the time $t$. 
In the case of $t=1$, we simply call it the lowest energy classical path from $y$ to $x$. 
\end{defi}

The object of this subsection is to show the existence and uniqueness of the lowest energy classical path connecting sufficiently close two points. 

\begin{theo}\label{theo_action_well_defined}
There exist a small $t_0>0$, a small neighborhood $N\subset M\times M$ of $\rmop{diag}(M)$ and $\mu\in\mathbb{R}$ such that  for all $(x, y)\in N$ and $t\in (0, t_0]$, the set $\Gamma^H_{t, x, y}(\mu/t)$
consists of a single element. 
\end{theo}

In order to prove Theorem \ref{theo_action_well_defined}, it is enough to prove the following local version. 

\begin{theo}\label{theo_low_energy}
Let $H: T^*M\to\mathbb{R}$ be a Hamiltonian of the form 
\[
H(x, \xi)=\frac{1}{2}|\xi|_g^2+V(x). 
\]
Then for all $y_0\in M$, there exist an open neighborhood $W$ of $y_0$, constants $\mu>0$ and $t_0^\prime>0$ such that for all $(x, y)\in W\times W$ and $t\in (0, t_0^\prime]$, the set $\Gamma^H_{t, x, y}(\mu/t)$
consists of a single element. 
\end{theo}

First we prove Theorem \ref{theo_action_well_defined} from Theorem \ref{theo_low_energy}. The rest of this subsection is devoted to the proof of Theorem \ref{theo_low_energy}. 

\begin{proof}[Proof of Theorem \ref{theo_action_well_defined}]
We introduce a family $\mathscr{U}$ of open sets in $M \times M$ as 
\begin{align*}
&\mathscr{U}:= \\
&\left\{\, 
W\times W \underset{\text{open}}{\subset}
M \times M 
\,\middle| \, 
\begin{aligned}
&\exists t_0^\prime>0 \text{ such that } 
\forall (t, x, y)\in (0, t_0^\prime)\times W \times W, \, \\
&\Gamma^H_{t, x, y} \text{ has a unique element with the lowest energy}
\end{aligned}
\,\right\}. 
\end{align*}

$\mathscr{U}$ is an open covering of $\rmop{diag}(M)$ by Theorem \ref{theo_low_energy}. Since $\rmop{diag} (M)$ is compact, it is covered by finite open sets $\{W_j\times W_j \in \mathscr{U}\}_{j=1}^J$. We take $t_{0, j}^\prime>0$ such that for all $(t, x, y)\in (0, t_{0, j}^\prime]\times W_j \times W_j$, $\Gamma^H_{t, x, y}$ has a unique element with the lowest energy. We set 
\[
t_0:=\min_{1\leq j \leq J} t_{0, j}^\prime, \, N=\bigcup_{j=1}^J (W_j\times W_j). 
\]
By definition, for all $(t, x, y)\in (0, t_0]\times N$, there exists a unique lowest energy classical path from $y$ to $x$ in time $t$. 
\end{proof}

In the following we prove Theorem \ref{theo_low_energy}. We begin with an observation. In the statement of Theorem \ref{theo_low_energy}, the energy bound $\mu/t$ diverges as $t\to +0$. Naively, the reason is that classical particle needs to move more quickly in order to go to a fixed point from another fixed point in shorter time $t$. This causes the divergence of the momentum of the classical path as $t\to +0$ and makes analysis of classical mechanics difficult. In order to avoid this problem, we introduce a rescaling on $T^*M$: 
\[
\Theta_t(x, \xi):=(x, t^{-1}\xi). 
\]
We also rescale the Hamiltonian as 
\[
H_t(x, \xi):=t^2H(x, t^{-1}\xi)=\frac{1}{2}|\xi|_g^2+t^2 V(x). 
\]
By the rescaling, the meaning of the small parameter $t>0$ is changed from the length of the time interval to the parameter of the perturbation by $t^2 V(x)$. In particular $H_1$ is the original Hamiltonian $H=|\xi|_g^2/2+V(x)$ and $H_0$ is the free Hamiltonian $H_0=|\xi|_g^2/2$. 
We denote by $X_{H_t}$ the Hamilton vector field of $H_t$ associated with the canonical symplectic structure $\omega=\sum dx_j\wedge d\xi_j$ on $T^*M$. Let $\{\varphi^t_s: T^*M\to T^*M\}_{s\in (-a, a)}$ be the flow of $X_H$: 
\begin{equation}\label{eq_hamiltonian_flow_scaled}
    \begin{split}
        \frac{\partial \varphi^t_s}{\partial s}(y, \eta)&=X_{H_t}(\varphi^t_s(y, \eta)), \\
        \varphi^t_0(y, \eta)&=(y, \eta). 
    \end{split}
\end{equation}
We denote $\varphi^1_s=\varphi_s$. $\{\varphi_s\}$ is the Hamiltonian flow of the original Hamiltonian $H$. The rescaling $\Theta_t$ relates the Hamiltonian flow $\{\varphi^t_s\}$ of the rescaled Hamiltonian $H_t$ to the Hamiltonian flow $\{\varphi_s\}$ of the original Hamiltonian $H$. 

\begin{lemm}\label{lemm_scaling}
If $\{\varphi^t_s: T^*M\to T^*M\}_{s\in (-a, a)}$ is the flow of $X_{H_t}$, then $\varphi_s$ is defined for $s\in (-at, at)$ and $\varphi_s=\Theta_t\circ \varphi^t_{t^{-1}s}\circ \Theta_t^{-1}$. 
\end{lemm}

\begin{proof}
First we prove that $\Theta_{t*}X_{H_t}=tX_H$ using local coordinates. Take local coordinates $(x_1, \ldots, x_n)$ of $M$ and the associated canonical coordinates $(x_1, \ldots, x_n, \xi_1, \ldots, \xi_n)$. Then
\[
(\Theta_{t*}X_{H_t})(x, \xi)=\sum_{j=1}^n\left(t\frac{\partial H}{\partial \xi_j}(x, \xi)\frac{\partial}{\partial x_j}
-t\frac{\partial H}{\partial x_j}(x, \xi)\frac{\partial}{\partial \xi_j}\right)
=tX_H(x, \xi). 
\]

We prove that $s\mapsto \Theta_t\circ\varphi^t_{t^{-1}s}\circ \Theta_t^{-1}(x, \xi)$ is the integral curve of $X_H$ for every $(x, \xi)\in T^*M$. This is proved by a direct calculation 
\begin{align*}
&\frac{d}{ds}(\Theta_t\circ \varphi^t_{t^{-1}s}\circ \Theta_t^{-1})(x, \xi)
=t^{-1}(d\Theta_t)_{(\varphi^t_{t^{-1}s}\circ \Theta_t^{-1})(x, \xi)}\left.\frac{d}{du}\right|_{u=t^{-1}s}\varphi^t_u(\Theta_t^{-1}(x, \xi)) \\
&=t^{-1}(d\Theta_t)_{(\varphi^t_{t^{-1}s}\circ \Theta_t^{-1})(x, \xi)}X_{H_t}((\varphi^t_{t^{-1}s}\circ\Theta_t^{-1})(x, \xi))
=t^{-1}(\Theta_{t*}X_{H_t})(x, \xi) \\
&=X_H(x, \xi). 
\end{align*}
Clearly $(\Theta_t\circ \varphi^t_{t^{-1}s}\circ \Theta_t^{-1}) (x, \xi)|_{s=0}=(x, \xi)$. Hence by the uniqueness of integral curve, we obtain $\varphi_s=\Theta_t\circ \varphi^t_{t^{-1}s}\circ \Theta_t^{-1}$ for $s\in (-at, at)$. 
\end{proof}

For $U\subset M$ and $\mu\geq 0$, we define $\mathcal{V}(U, \mu)\subset T^*M$ as 
\[
\mathcal{V}(U, \mu):=\{\, (x, \xi)\in T^*U \mid x\in U, \, |\xi|_g<\mu\,\}. 
\]

If we regard $t$ as a parameter of the perturbation, $H_t(x, \xi)$ also makes sense for $t\leq 0$, while $t$ as the length of the time interval ($t$ in $\Gamma^H_{t, x, y}$) makes sense only for $t>0$. Moreover the rescaled Hamiltonian $H_0$ corresponds to the free Hamiltonian $|\xi|_g^2/2$. Thus we can analyze the classical mechanics with the rescaled Hamiltonian $H_t$ employing the calculus around $t=0$. 

\begin{lemm}\label{lemm_low_energy_diffeo}
There exist $t_1>0$ and $\mu_0>0$ such that 
\[
\varphi^t_1: \mathcal{V}(M, \mu_0)\longrightarrow \varphi^t_1(\mathcal{V}(M, \mu_0))
\]
is a diffeomorphism for all $t\in [-t_1, t_1]$. 
\end{lemm}

\begin{proof}
Let $y_0\in M$. Since $(0, y, 0)\in \mathbb{R}\times T^*M$ is a regular point of the mapping 
\[
F(t, y, \eta):=(t, \varphi^t_1(y, \eta))
\]
defined near $(0, y_0, 0)$, we can apply the inverse function theorem. Thus there exist $t_1>0$, a neighborhood $U$ of $y_0$ and $\mu>0$ such that 
\[
\varphi^t_1: \mathcal{V}(U, \mu)\longrightarrow \varphi^t_1(\mathcal{V}(U, \mu))
\]
is a diffeomorphism for all $t\in (-t_1, t_1)$. 

Now let $\mathscr{U}$ be a family of open subsets defined as
\[
\mathscr{U}:=\left\{\, (-t_1, t_1)\times \mathcal{V}(U, \mu)\,\middle| \, 
\begin{aligned}
&\varphi^t_1: \mathcal{V}(U, \mu)\to \varphi^t_1(\mathcal{V}(U, \mu)) \text{ is } \\
&\text{a diffeomorphism for all } t\in [-t_1, t_1]
\end{aligned}
\,\right\}. 
\]
This $\mathscr{U}$ covers a compact set $\{0\}\times \{\,(x, 0)\in T^*M \mid x\in M\,\}$. Thus we can choose finite sets $\{ (-t_{1, j}, t_{1, j})\times \mathcal{V}(U_j, \mu_j)\}_{j=1}^J$ which cover $\{0\}\times \{\,(x, 0)\in T^*M \mid x\in M\,\}$. We set $t_1:=\min_{1\leq j\leq J} t_{2, j}$ and $\mu_0:=\min_{1\leq j\leq J} \mu_j$. Then 
\[
\varphi^t_1: \mathcal{V}(M, \mu_0)\longrightarrow \varphi^t_1(\mathcal{V}(M, \mu_0))
\]
is a diffeomorphism for all $t\in [-t_1, t_1]$. 
\end{proof}

In the following we take $\mu_0>0$ as in Lemma \ref{lemm_low_energy_diffeo}. The next proposition is a rescaled counterpart of Theorem \ref{theo_low_energy}. 

\begin{prop}\label{prop_low_energy_scaled}
Let $y_0\in M$. Then there exist an open neighborhood $W$ of $y_0$, positive constants $\mu\in (0, \mu_0]$ and $t_0^\prime>0$ such that for all $(x, y)\in W\times W$ and $t\in [-t_0^\prime, t_0^\prime]$, the set $\Gamma^{H_t}_{1, x, y}(\mu)$ consists of a single element. 
\end{prop}

Note that the length of the time interval is fixed (always $[0, 1]$) and the divergence of the energy bound $\mu$ is removed by the rescaling. 

\begin{proof}
Let $\varphi^t_s(y, \eta)=(\overline{q}^t_s(y, \eta),\overline{p}^t_s(y, \eta))$. Then the mapping 
\[
\Lambda: (t, y, \eta)\longmapsto (t, \overline{q}^t_1(y, \eta), y). 
\]
is defined near $(0, y_0, 0)\in \mathbb{R}\times T^*M$. 
Take canonical coordinates $(y_1, \ldots, y_n, \eta_1, \ldots, \eta_n)$ near $(y_0, 0)$. Then the Jacobian of $\Lambda$ at $(0, y_0, 0)$ is
\[
\det J\Lambda(0, y_0, 0)
=\det
\begin{pmatrix}
1 & \partial_t \overline{q}^t_1(y_0, 0) & 0 \\
0 & \partial_y \overline{q}^t_1(y_0, 0) & 1 \\
0 & \partial_\eta \overline{q}^t_1(y_0, 0) & 0
\end{pmatrix}
=(-1)^n\det \partial_\eta \overline{q}^t_1(y_0, 0). 
\]
This is not 0 since $(y, \eta)\mapsto (\overline{q}^0_1(y, \eta), y)=(\exp_y(v_\eta), y)$ is a local diffeomorphism around $(y_0, 0)$ by the existence of geodesically convex neighborhoods\footnote{See \cite{Spivak79} for the existence of geodesically convex neighborhoods. }. Here $v_\eta\in T_yM$ is a tangent vector defined as $g(v_\eta, \cdot)=\eta$. 
By the inverse function theorem, there exist $t_0^\prime>0$ and a neighborhood $\mathcal{V}=\mathcal{V}(U, \mu)$ of $(y_0, 0)$ such that 
\[
\tilde\Lambda:=\Lambda|_{[-t_0^\prime, t_0^\prime]\times \mathcal{V}}: [-t_0^\prime, t_0^\prime]\times \mathcal{V}\longrightarrow \Omega:=\Lambda([-t_0^\prime, t_0^\prime]\times \mathcal{V})
\]
is a diffeomorphism. Take $\mu_0$ as in Lemma \ref{lemm_low_energy_diffeo}. We assume that $\mu\in (0, \mu_0]$. 

The inverse function $\tilde\Lambda^{-1}$ is of the form 
\[
\tilde \Lambda^{-1}(t, x, y)=(t, \zeta(t, x, y))
\]
for some smooth function $\zeta: \Omega\to \mathcal{V}$. Since $(0, y_0, y_0)\in \Omega$, we can take a neighborhood $W$ of $y_0$ such that $[-t_0^\prime, t_0^\prime]\times W\times W\subset \Omega$. Now we define
\[
\gamma^t_s(x, y):=\varphi^t_s(\zeta(t, x, y))
\]
for $(s, t, x, y)\in [-1, 1]\times [-t_0^\prime, t_0^\prime]\times W\times W$. This curve is an element in $\Gamma^{H_t}_{1, x, y}$ by definition. 

Finally we prove that the above $\gamma^t_s(x, y)=(q^t_s(x, y), p^t_s(x, y))$ is the unique element of the set $\Gamma^{H_t}_{t, x, y}(\mu)$. 
Let $\alpha(s)=(q(s), p(s))\in \Gamma^{H_t}_{1, x, y}$ have an initial momentum such that $|p(0)|_g<\mu$. Since $q(1)=\overline{q}^t_1(y, p(0))=x$, we have
\[
(t, y, p(0))=\tilde\Lambda^{-1}(t, \overline{q}^t_1(y, p(0)), y)=\tilde\Lambda^{-1}(t, x, y)=(t, y, p^t_0(x, y)). 
\]
The uniqueness of the solution to initial value problem implies that $\alpha(s)=\gamma^t_s(x, y)$ for all $s\in [0, 1]$. 
\end{proof}

We prove Theorem \ref{theo_low_energy} by converting Proposition \ref{prop_low_energy_scaled} using Lemma \ref{lemm_scaling}. 

\begin{proof}[Proof of Theorem \ref{theo_low_energy}]
    Let $y_0\in M$. Then by Propotision \ref{prop_low_energy_scaled}, there exist an open neighborhood $W$ of $y_0$, positive constants $\mu\in (0, \mu_0]$ and $t_0^\prime>0$ such that for all $(x, y)\in W\times W$ and $t\in [-t_0^\prime, t_0^\prime]$, the set $\Gamma^{H_t}_{1, x, y}(\mu)$ has a unique element. Let $(q^t_s (x, y), p^t_s(x, y))\in \Gamma^{H_t}_{1, x, y}(\mu)$ be the unique element. We define 
    \begin{equation}\label{eq_element_scale}
        (x^t_s(x, y), \xi^t_s(x, y)):=\Theta_t (q^t_{t^{-1}s}(x, y), p^t_{t^{-1}s}(x, y))=(q^t_{t^{-1}s}(x, y), t^{-1}p^t_{t^{-1}s}(x, y))
    \end{equation}
    for $(x, y)\in W\times W$, $t\in (0, t_0^\prime]$ and $s\in [-t, t]$. We claim that $(x^t_s(x, y), \xi^t_s(x, y))$ is the unique element of $\Gamma^H_{t, x, y}(\mu/t)$. $x^t_0(x, y)=y$ and $x^t_t(x, y)=x$ are immediately obtained by the definition \eqref{eq_element_scale}. We put $\eta:=\xi^t_0(x, y)=t^{-1}p^t_0 (x, y)$. Then, since $(q^t_s(x, y), p^t_s(x, y))$ is a classical path with respect to the Hamiltonian $H_t$ with initial value $(y, t\eta)$, we have 
    \begin{equation}\label{eq_qp_yeta}
        (q^t_{t^{-1}s}(x, y), p^t_{t^{-1}s}(x, y))=\varphi^t_{t^{-1}s}(y, t\eta)=(\varphi^t_{t^{-1}s}\circ \Theta_t^{-1})(y, \eta). 
    \end{equation}
    By applying $\Theta_t$ both side of \eqref{eq_qp_yeta} and recalling the definition \eqref{eq_element_scale} of $(x^t_s, \xi^t_s)$, we obtain 
    \[(x^t_s(x, y), \xi^t_s(x, y))=\Theta_t (q^t_{t^{-1}s}(x, y), p^t_{t^{-1}s}(x, y))=(\Theta_t\circ \varphi^t_{t^{-1}s}\circ \Theta_t^{-1})(y, \eta). 
        \]
    Now we can apply Lemma \ref{lemm_scaling} and obtain
    \begin{equation}\label{eq_xxi_phi}
        (x^t_s(x, y), \xi^t_s(x, y))=\varphi_s(y, \eta). 
    \end{equation}
    Thus $(x^t_s(x, y), \xi^t_s(x, y))$ is a classical path with respect to the Hamiltonian $H$. The initial momentum $\xi^t_0(x, y)$ is estimated as
    \[ |\xi^t_0(x, y)|_g =|p^t_0 (x, y)|_g/t < \mu /t\]
    since $\xi^t_0(x, y)=p^t_0(x, y)$ by \eqref{eq_element_scale} and $|p^t_0(x, y)|_g< \mu$ by the definition of $\Gamma^{H_t}_{1, x, y}(\mu)$. Hence we proved $(x^t_s(x, y), \xi^t_s(x, y))\in \Gamma^H_{t, x, y}(\mu/t)$. 

    For the proof of the uniqueness, we assume $(\tilde x(s), \tilde \xi(s))\in \Gamma^H_{t, x, y}(\mu/t)$. Consider 
    \begin{equation}\label{eq_element_scaling_inv}
        (\tilde q(s), \tilde p(s)):=\Theta_t^{-1}(\tilde x(ts), \tilde \xi(ts))=(\tilde x(ts), t\tilde \xi (ts))
    \end{equation}
    for $s\in [-1, 1]$. We claim that $(\tilde q(s), \tilde p(s))\in \Gamma^{H_t}_{1, x, y}(\mu)$. We easily obtain $\tilde q(0)=\tilde x(0)=y$, $\tilde q(1)=\tilde x(t)=x$ and $|\tilde p(0)|_g=t|\tilde \xi(0)|_g<\mu$ by $(\tilde x(s), \tilde \xi(s))\in \Gamma^H_{t, x, y}(\mu/t)$. What we have to prove is that $(\tilde q(s), \tilde p(s))$ is a classical path with respect to the Hamiltonian $H_t$. Since $(\tilde x(s), \tilde \xi(s))\in \Gamma^H_{t, x, y}(\mu/t)$, we have 
    \[(\tilde x(ts), \tilde \xi(ts))=\varphi_{ts}(y, \tilde \xi(0))=\varphi_{ts}(y, t^{-1}\tilde p(0))=(\varphi_{ts}\circ\Theta_t)(y, \tilde p(0))\] 
    for all $s\in [-1, 1]$. Applying $\Theta_t^{-1}$ to the both sides and recalling the definition \eqref{eq_element_scaling_inv}, we obtain 
    \begin{equation}\label{eq_xxi_scaling_inv}
        (\tilde q(s), \tilde p(s))=\Theta_t^{-1}(\tilde x(ts), \tilde \xi(ts))=(\Theta_t^{-1}\circ \varphi_{ts}\circ\Theta_t)(y, \tilde p(0)). 
    \end{equation}
    We can apply Lemma \ref{lemm_scaling} to \eqref{eq_xxi_scaling_inv} and obtain 
    \[ (\tilde q(s), \tilde p(s))=\varphi^t_s(y, \tilde p(0)). \]
    Hence we proved that $(\tilde q(s), \tilde p(s))$ is a classical path with respect to the Hamiltonian $H_t$. Therefore $(\tilde q(s), \tilde p(s))\in \Gamma^{H_t}_{1, x, y}(\mu)$. By the uniqueness of the element of $\Gamma^{H_t}_{1, x, y}(\mu)$, which is the assertion of Proposition \ref{prop_low_energy_scaled}, we obtain $(\tilde q(s), \tilde p(s))=(q^t_s(x, y), p^t_s(x, y))$. 
    Hence
    \begin{align*}
     (\tilde x(s), \tilde \xi(s))
     &=(\tilde q(t^{-1}s), t^{-1}\tilde p(t^{-1}s))=(q^t_{t^{-1}s}(x, y), t^{-1}p^t_{t^{-1}s}(x, y)) \\
     &=(x^t_s(x, y), \xi^t_s(x, y))
    \end{align*}
by \eqref{eq_element_scale} and \eqref{eq_element_scaling_inv}. 
\end{proof}

\subsection{Action integrals}

We recall the definition of the action integral. 

\begin{defi}
Take $t_0>0$ and $N\subset M\times M$ as in Theorem \ref{theo_action_well_defined}. For $(t, x, y)\in (0, t_0]\times N$, let $(x^t_s(x, y), \xi^t_s(x, y))\in \Gamma^H_{t, x, y}$ be the lowest energy classical path from $y$ to $x$ in the time $t$. Then we define 
\[
S(t, x, y):=\int_0^t \left(\frac{1}{2}\left|\frac{d}{ds}x^t_s(x, y)\right|_g^2-V(x^t_s(x, y))\right) \,ds. 
\]
\end{defi}

$S(t, x, y)$ has another representation 
\begin{equation} S(t, x, y)=\int_0^t \left(\jbracket{\xi^t_s(x, y), \frac{d}{ds} x^t_s(x, y)}-H(x^t_s(x, y), \xi^t_s(x, y))\right)\, ds. \label{eq_action_hamiltonian}\end{equation}
This follows from 
\[ g\left( \frac{d}{ds}x^t_s(x, y), \cdot\right)=\xi^t_s (x, y)\]
by the Hamilton equation and hence 
\begin{align*}\frac{1}{2}\left|\frac{d}{ds}x^t_s\right|_g^2-V(x^t_s)
    &=\frac{1}{2}\left|\frac{d}{ds}x^t_s\right|_g^2+\frac{1}{2}\left| \xi^t_s\right|_g^2-H(x^t_s, \xi^t_s) \\ 
    &=\jbracket{\xi^t_s, \frac{d}{ds} x^t_s}-H(x^t_s, \xi^t_s). 
\end{align*}

We recall the well known facts on the action integrals in the classical mechanics. As in Subsection \ref{subsec_classical_mechanics}, let $\{\varphi_s: T^*M\to T^*M\}_{s\in (-a, a)}$ be the flow generated by the Hamilton vector field $X_H$ associated with the Hamiltonian $H$. 

\begin{theo}\label{theo_property_of_action}
\begin{enumerate}
\item $S(t, x, y)$ generates the symplectomorphism $\varphi_t$. More precisely, for $(t, x, y)\in (0, t_0]\times N$, the relation 
\[
\varphi_t (y, -d_y S(t, x, y))=(x, d_x S(t, x, y))
\]
holds. 
\item $S(t, x, y)$ satisfies the Hamilton-Jacobi equation
\begin{equation}
\frac{\partial S}{\partial t}(t, x, y)+H( x, d_x S(t, x, y))=0. \label{eq_hj}
\end{equation}
\end{enumerate}
\end{theo}

\begin{proof}
    Let $(x^t_s(x, y), \xi^t_s(x, y))\in \Gamma^H_{t, x, y}$ be the lowest energy classical path. Note that 
    \begin{equation}\varphi_t(y, \xi^t_t(x, y))=(x, \xi^t_0(x, y)) \label{eq_hamilton_flow_split}\end{equation}
    by the definition of $\varphi_t$ and $(x^t_s, \xi^t_s)$. Take arbitrary local coordinates $\psi=(x_1, \ldots, x_n)$ and lift it to the  canonical coordinates on $T^*M$: 
    \begin{equation}
        \tilde \psi: (x, \xi_1 dx_1+\cdots +\xi_n dx_n) \longmapsto (\psi(x), \xi_1, \ldots, \xi_n). \label{eq_canonical_coordinates}\end{equation}
        We define $x^t_{s, j}(x, y)\in \mathbb{R}$ and $\xi^t_{s, j}(x, y)\in \mathbb{R}$ as 
    \begin{equation}\label{eq_x_xi_local}
        \begin{split}
        \psi (x^t_s(x, y))&=(x^t_{s, 1}(x, y), \ldots , x^t_{s, n}(x, y)), \\
    \xi^t_s(x, y)&=\xi^t_{s, 1}(x, y)dx_1+\cdots + \xi^t_{s, n}(x, y)dx_n. 
        \end{split}
    \end{equation}
    $x^t_{s, i}$ and $\xi^t_{s, i}$ satisfies the Hamilton equations
    \begin{equation}
            \frac{dx^t_{s, j}}{ds} =\frac{\partial H}{\partial \xi_j}(x^t_s, \xi^t_s), \quad \frac{d\xi^t_{s, j}}{ds} =-\frac{\partial H}{\partial x_j}(x^t_s, \xi^t_s) \label{eq_hamilton}
    \end{equation}
    In the following, we use a shorthand notation $\dot{}=\partial_s$. For example we denote $\dot x^t_s=\partial_s x^t_s$ and $\dot \xi^t_s=\partial_s \xi^t_s$. 
    
    1. By \eqref{eq_hamilton_flow_split}, it suffices to show that 
    \begin{equation}
        \xi^t_t(x, y)=d_x S(t, x, y),  \xi^t_0(x, y)=-d_y S(t, x, y). \label{eq_generating}
    \end{equation}
    We claim that 
    \begin{equation}d_x (\jbracket{\xi^t_s, \dot x^t_s}-H(x^t_s, \xi^t_s))
    =\frac{\partial}{\partial s}(d_x x^t_s(\cdot, y)^* \xi^t_s). \label{eq_generating_locally}\end{equation}
    Here $d_x x^t_s (\cdot, y)^*: T^*_{x^t_s(x, y)}M\to T^*_xM$ is the dual mapping of the derivative $d_x x^t_s (\cdot, y): T_xM\to T_{x^t_s(x, y)}M$ of the smooth mapping $x\mapsto x^t_s(x, y)$. In the canonical coordinates \eqref{eq_canonical_coordinates}, we have
    \begin{align*} &\frac{\partial}{\partial x_j}(\jbracket{\xi^t_s, \dot x^t_s}-H(x^t_s, \xi^t_s)) \\
        &=\sum_{k=1}^n\left(\frac{\partial \xi^t_{s, k}}{\partial x_j} \dot x^t_{s, k}+\xi^t_{s, k} \frac{\partial^2 x^t_{s, k}}{\partial s \partial x_j}-\frac{\partial H}{\partial x_k}\frac{\partial x^t_{s, k}}{\partial x_j}-\frac{\partial H}{\partial \xi_k}\frac{\partial \xi^t_{s, k}}{\partial x_j}\right) \\
    &= \sum_{k=1}^n \Biggl(\frac{\partial \xi^t_{s, k}}{\partial x_j}\underbrace{\left(\dot x^t_{s, k}-\frac{\partial H}{\partial \xi_k}\right)}_{=0} -\frac{\partial x^t_{s, k}}{\partial x_j}\underbrace{\left(\dot \xi^t_{s, k}+\frac{\partial H}{\partial x_k}\right)}_{=0}+\frac{\partial}{\partial s}\left(\xi^t_{s, k}\frac{\partial x^t_{s, k}}{\partial x_j}\right) \Biggr)\end{align*}
    by the Hamilton equations \eqref{eq_hamilton}. Hence 
    \begin{equation}
        \begin{split}
            d_x (\jbracket{\xi^t_s, \dot x^t_s}-H(x^t_s, \xi^t_s))
            &=\sum_{j=1}^n \frac{\partial}{\partial x_j}(\jbracket{\xi^t_s, \dot x^t_s}-H(x^t_s, \xi^t_s))dx_j \\
            &=\frac{\partial}{\partial s}\left(\sum_{j, k=1}^n\xi^t_{s, k}\frac{\partial x^t_{s, k}}{\partial x_j}dx_j\right). 
        \end{split}
        \label{eq_generating_locally_wip}
    \end{equation}
    $\sum_{j, k=1}^n \xi^t_{s, k}\partial_{x_j}x^t_{s, k}dx_j$ is just a local representation of $d_x x^t_s(\cdot, y)^*\xi^t_s$: 
    \[ \sum_{j, k=1}^n\xi^t_{s, k}\frac{\partial x^t_{s, k}}{\partial x_j}dx_j=d_x x^t_s(\cdot, y)^*\xi^t_s. \]
    Hence \eqref{eq_generating_locally_wip} implies \eqref{eq_generating_locally}. By \eqref{eq_generating_locally}, we obtain  
    \begin{align*}
        d_xS(t, x, y)&=\int_0^t d_x (\jbracket{\xi^t_s, \dot x^t_s}-H(x^t_s, \xi^t_s))\, ds 
        =\int_0^t \frac{\partial}{\partial s}(d_x x^t_s(\cdot, y)^* \xi^t_s)\, ds \\
        &=d_x x^t_t (x, y)^*\xi^t_t(x, y) -d_x x^t_0 (x, y)^*\xi^t_0 (x, y)=\xi^t_t(x, y)
    \end{align*}
    since $x^t_t(x, y)=x$ and $x^t_0(x, y)=y$ imply that $d_x x^t_t(x, y)=\rmop{id}\nolimits_{T_xM}$ and $d_x x^t_0 (x, y)=0$. 

    The relation $\xi^t_0(x, y)=-d_yS(t, x, y)$ is proved similarly and we omit the proof. 

    2. Let \eqref{eq_canonical_coordinates} be canonical coordinates near a fixed point $(x^t_s(x, y), \xi^t_s(x, y))\in T^*M$ and we use the same notation $x^t_{s, j}, \xi^t_{s, j}$ as in \eqref{eq_x_xi_local}. 
    Then the derivative by $t$ of the integrand of \eqref{eq_action_hamiltonian} is 
    \begin{align*}
        &\frac{\partial}{\partial t}(\jbracket{\xi^t_s, \dot x^t_s}-H(x^t_s, \xi^t_s)) \\
        &=\sum_{j=1}^n \left(\frac{\partial \xi^t_{s, j}}{\partial t}\dot x^t_{s, j}+\xi^t_{s, j}\frac{\partial^2 x^t_{s, j}}{\partial s \partial t}-\frac{\partial H}{\partial x_j}\frac{\partial x^t_{s, j}}{\partial t}-\frac{\partial H}{\partial \xi_j}\frac{\partial \xi^t_{s, j}}{\partial t}\right) \\
        &=\sum_{j=1}^n \Biggl( \frac{\partial \xi^t_{s, j}}{\partial t}\underbrace{\left(\dot x^t_{s, j}-\frac{\partial H}{\partial \xi_j}\right)}_{=0}-\frac{\partial x^t_{s, j}}{\partial t}\underbrace{\left( \dot \xi^t_{s, j}+\frac{\partial H}{\partial x_j}\right)}_{=0}+\frac{\partial}{\partial s}\left(\xi^t_{s, j}\frac{\partial x^t_{s, j}}{\partial t}\right)\Biggr) \\
        &= \frac{\partial}{\partial s}\left(\sum_{j=1}^n \xi^t_{s, j}\frac{\partial x^t_{s, j}}{\partial t}\right)=\frac{\partial}{\partial s}\jbracket{\xi^t_s, \frac{\partial x^t_s}{\partial t}}
    \end{align*}
    by the Hamilton equations \eqref{eq_hamilton}. Hence we obtain 
    \begin{align}
        \frac{\partial S}{\partial t}(t, x, y)
        &= \jbracket{\xi^t_t, \dot x^t_t}-H(x^t_t, \xi^t_t)+\int_0^t \frac{\partial}{\partial t}(\jbracket{\xi^t_s, \dot x^t_s}-H(x^t_s, \xi^t_s))\, ds \nonumber \\
        &=\jbracket{\xi^t_t, \dot x^t_t}-H(x, \xi^t_t)+\int_0^t \frac{\partial}{\partial s}\jbracket{\xi^t_s, \frac{\partial x^t_s}{\partial t}}\, ds \nonumber \\
        \begin{split}&=\jbracket{\xi^t_t, \dot x^t_t}-H(x, \xi^t_t) \\
        &\quad +\jbracket{\xi^t_t(x, y), \left.\frac{\partial x^t_s}{\partial t}(x, y)\right|_{s=t}}-\jbracket{\xi^t_0(x, y), \frac{\partial x^t_0}{\partial t}(x, y)}. \end{split}\label{eq_time_derivative_of_action} 
    \end{align}

    We calculate tangent vectors $\partial_t x^t_s(x, y)|_{s=t}\in T_x M$ and $\partial_t x^t_0(x, y)\in T_yM$. $\partial_t x^t_0(x, y)=0$ is easily proved by $x^t_0(x, y)=y$ for any $t$. $\partial_t x^t_s(x, y)|_{s=t}$ is related to a velocity vector by differentiating the identity $x^t_t(x, y)=x$ by $t$: 
    \begin{equation}\label{eq_perturbation_velocity}
        0=\frac{\partial}{\partial t} x^t_t (x, y)=\left.\frac{\partial x^t_s}{\partial t}(x, y)\right|_{s=t}+\dot x^t_t(x, y). 
    \end{equation} 
    Hence \eqref{eq_time_derivative_of_action}, \eqref{eq_perturbation_velocity} and $\partial_t x^t_0(x, y)=0$ imply 
    \[
        \frac{\partial S}{\partial t}(t, x, y)+H(x, \xi^t_t(x, y))=\jbracket{\xi^t_t(x, y), \dot x^t_t(x, y)+\left.\frac{\partial x^t_s}{\partial t}(x, y)\right|_{s=t}}=0. 
    \]
    Combining this equation with $d_xS(t, x, y)=\xi^t_t(x, y)$ of \eqref{eq_generating}, we obtain the desired equation \eqref{eq_hj}. 
\end{proof}

Similarly to the problem of the divergence of the momentum mentioned in Subsection \ref{subsec_classical_mechanics}, the action $S(t, x, y)$ diverges as $t\to +0$. We can observe the divergence by considering the free particle: $S(t, x, y)=d(x, y)^2/2t$, where $d(x, y)$ is the distance between two points $x$ and $y$ in $M$ associated with the Riemannian metric $g$. In order to pick up the order of divergence, we introduce another action function, which has a better analytic properties than $S(t, x, y)$ and has an easy relation to it. Fix $t_0>0$ and $N\subset M\times M$ as in Theorem \ref{theo_action_well_defined}. For $(t, x, y)\in [-t_0, t_0]\times N$, let $(q^t_s(x, y), p^t_s(x, y))\in \Gamma^{H_t}_{1, x, y}$ be the unique lowest energy classical path from $y$ to $x$ with respect to the rescaled Hamiltonian $H_t=t^2H\circ\Theta_t$. We define the action $\Phi(t, x, y)$ of $(q^t_s(x, y), p^t_s(x, y))$ as 
\[
\Phi(t, x, y):=\int_0^1 \left(\frac{1}{2}\left|\frac{d}{ds}q^t_s(x, y)\right|_g^2-t^2V(q^t_s(x, y))\right) \,ds. 
\]

Note that $\Phi(t, x, y)$ is a smooth function not only in $(0, t_0]\times N$, but also in $[-t_0, t_0]\times N$ by the smoothness of the solutions of the initial value problem of ordinary differential equations with respect to initial values and parameters. We use the same character $t$ in $S(t, x, y)$ and $\Phi(t, x, y)$, though the meaning of $t$ is different. $t\in (0, t_0]$ in $S(t, x, y)$ is the length of the time interval, while $t\in [-t_0, t_0]$ in $\Phi(t, x, y)$ is a parameter of the perturbation by the potential $t^2V(x)$. 

There exists an easy relation between $S(t, x, y)$ and $\Phi(t, x, y)$. 

\begin{theo}\label{theo_action_scaling}
We have
\[
S(t, x, y)=t^{-1}\Phi(t, x, y)
\]
for all $(t, x, y)\in (0, t_0]\times N$. 
\end{theo}

\begin{proof}
Let $(x^t_s(x, y), \xi^t_s(x, y))\in \Gamma^H_{t, x, y}$ be the lowest energy path with respect to $H$ and $(q^t_s(x, y), p^t_s(x, y))\in \Gamma^{H_t}_{1, x, y}$ be that with respect to $H_t$. Take $\eta=\xi^t_0(x, y)\in T^*_yM$. Then $(x^t_s(x, y), \xi^t_s(x, y))=\varphi_s(y, \eta)$. Let $(\pi\circ\varphi_s)(y, \eta)=\tilde x_s(y, \eta)$. We also put $\eta^\prime:=p^t_0(x, y)$ and $\gamma^t_s(x, y)=\varphi^t_{t^{-1}s}(y, \eta^\prime)$, $(\pi\circ \varphi^t_{t^{-1}s})(y, \eta^\prime)=\tilde q^t_s(y, \eta^\prime)$. By Lemma \ref{lemm_scaling}, we obtain
\[
\tilde x_s(y, \eta)=\tilde q^t_{t^{-1}s}(y, t\eta). 
\]
Substitute $s=t$ to this and we have $x=\tilde q^t_1(y, t\eta)$. This implies that $p^t_0(x, y)=t\eta$. By the definition of $\eta^\prime$, we obtain $\eta^\prime=t\eta$. Hence $x^t_s(x, y)=q^t_{t^{-1}s}(x, y)$. 

We calculate $S(t, x, y)$ as 
\begin{align*}
S(t, x, y)
&=\int_0^t \left(\frac{1}{2}\left|\frac{d}{ds}x^t_s(x, y)\right|_g^2-V(x^t_s(x, y))\right) \,ds \\
&=\int_0^t \left(\frac{1}{2}t^{-2}\left|\left.\frac{dq^t_\sigma}{d\sigma}\right|_{\sigma=t^{-1}s}(x, y)\right|_g^2-V(q^t_{t^{-1}s}(x, y))\right) \,ds \\
&=t^{-1}\int_0^1 \left(\frac{1}{2}\left|\frac{dq^t_s}{ds}(x, y)\right|_g^2-t^2V(q^t_s(x, y))\right) \,ds \\
&=t^{-1}\Phi(t, x, y). \qedhere
\end{align*}
\end{proof}

The function $\Phi(t, x, y)$ plays an important role in Section \ref{sect_FIO} due to the following properties. 

\begin{theo}\label{theo_phase_function}
    Take $\mu_0>0$ as in Lemma \ref{lemm_low_energy_diffeo}. If we take smaller $t_0>0$ if necessary, then the following statements hold. 
    \begin{enumerate}
    \item $(t, y, \eta)\mapsto \varphi^t_1(y, \eta)$ is smooth; 
    \item $\varphi^t_1: \mathcal{V}(M, \mu_0)\to \varphi^t_1(\mathcal{V}(M, \mu_0))$ is a diffeomorphism for each $t\in [-t_0, t_0]$; 
    \item If $(t, x, y)\in [-t_0, t_0]\times N$, then $(y, -d_y\Phi (t, x, y))\in \mathcal{V}(M, \mu_0)$ and 
    \[
    \varphi^t_1(y, -d_y\Phi (t, x, y))=(x, d_x\Phi (t, x, y)). 
    \]
    \end{enumerate}
    \end{theo}

\begin{proof}
We introduce a shorthand notation $\mathcal{V}=\mathcal{V}(M, \mu_0)$. We take $t_1>0$ as in Lemma \ref{lemm_low_energy_diffeo}. Take smaller $t_0$ if necessary, we can assume that $t_0=t_1$. Consider the family of diffeomorphisms 
\[
\{\varphi^t_1: \mathcal{V}\to \varphi^t_1(\mathcal{V})\}_{t\in [-t_0, t_0]}
\]
generated by $\Phi$. 
This obviously satisfies the assertions 1 and 2 in Theorem \ref{theo_phase_function}. We prove the assertion 3. By Theorem \ref{theo_action_well_defined}, we can take the unique minimizer $(x^t_s(x, y), \xi^t_s(x, y))\in \Gamma^H_{t, x, y}(\mu_0/t)$ of initial energy in $\Gamma^H_{t, x, y}(\mu_0/t)$. Then, since $S(t, x, y)$ generates $(x^t_s, \xi^t_s)$ by Theorem \ref{theo_property_of_action}, we obtain
\[
(y, -d_yS(t, x, y))=(y, \xi^t_0(x, y))\in \mathcal{V}(M, \mu_0/t)
\]
and $\varphi_t(y, -d_y S(t, x, y))=(x, d_xS(t, x, y))$. We recall that $\varphi_t=\Theta_t \circ \varphi^t_1 \circ \Theta_t^{-1}$ by Lemma \ref{lemm_scaling} and $\Phi(t, x, y)=t^{-1}S(t, x, y)$ by Theorem \ref{theo_action_scaling}. Then 
\[
|-d_y\Phi (t, x, y)|_g=t |d_y S(t, x, y)|_g< t \times \mu_0/t=\mu_0, 
\]
which implies $(y, -d_y\Phi(t, x, y))\in \mathcal{V}(M, \mu_0)$, and 
\begin{align*} \varphi^t_1 (y, -d_y \Phi(t, x, y))
    &=(\Theta_t^{-1} \circ \varphi_t)(y, -d_y S(t, x, y)) 
    =(y, t\xi^t_0(x, y)) \\
    &=(y, td_x S(t, x, y))=(y, d_x \Phi(t, x, y)). \qedhere\end{align*}
\end{proof}

In the end of this chapter, we record the difference of the action $\Phi(t, x, y)$ from the free motion. 
Let $d: M\times M \to [0, \infty)$ be the distance function associated with the Riemann matric $g$. 

\begin{theo}\label{theo_difference_from_free}
Let $\phi: [-t_0, t_0]\times N\to \mathbb{R}$ be a smooth function defined as 
\[
\phi(t, x, y):=\int_0^1 (1-s)(\partial_t^2\Phi)(st, x, y)\, ds. 
\]
Then $\Phi (t, x, y)=d(x, y)^2/2+t^2\phi(t, x, y)$. 
\end{theo}

\begin{proof}
Since $H_t(x, \xi)=H_{-t}(x, \xi)$ for all $(t, x, \xi)\in \mathbb{R}\times T^*M$, it follows that $q^t_s(x, y)=q^{-t}_s(x, y)$ for all $(t, x, y)\in [-t_0, t_0]\times N$ and $-1 \leq s\leq 1$. Hence $\Phi(t, x, y)=\Phi(-t, x, y)$ for all $(t, x, y)\in [-t_0, t_0]\times N$. This means that $\partial_t \Phi(0, x, y)=0$. Combining this with $\Phi(0, x, y)=d(x, y)^2/2$, we obtain the conclusion. 
\end{proof}

\subsection{Properties of Morette-Van Vleck determinants}\label{subsec_mvv}

We summarize the properties of Morette-Van Vleck determinants. 

\subsubsection{The value on the diagonal}\label{subsubsec_action}

As in the definition of the Morette-Van Vleck determinant \eqref{eq_morette_van_vleck}, we define 
\begin{equation}\label{eq_definiton_mvv_scaled}
    D_\Phi(t, x, y):= g_\lambda (x)^{-1/2} g_\mu (y)^{-1/2} \det \left(-\frac{\partial^2 \Phi}{\partial x_j\partial y_k}(t, x, y)\right)_{j, k=1}^n. \end{equation}
Since $\Phi(t, x, y)=tS(t, x, y)$, we have the relation $D_\Phi (t, x, y)=t^n D(t, x, y)$. The value of the Morette-Van Vleck determinant on the diagonal is important for the short time behavior of the short-time approximate solution $E_\hbar (t)$ (Proposition \ref{prop_fundamental_properties} (ii)). 

\begin{prop}\label{prop_amplitude_diagonal}
    We have $D_\Phi (0, x, x)=1$ for all $x\in M$. 
    \end{prop}
    
    \begin{proof}
    The action $S(t, x, y)=t^{-1}\Phi(t, x, y)$ satisfies the Hamilton-Jacobi equation
    \[
    \frac{\partial S}{\partial t}(t, x, y)+\frac{1}{2}|\rmop{grad}\nolimits_x S(t, x, y)|_g^2+V(x)=0
    \]
    by Theorem \ref{theo_property_of_action}. 
    Thus $\Phi(t, x, y)$ satisfies
    \[
    -\Phi(t, x, y)+t\frac{\partial \Phi}{\partial t}(t, x, y)+\frac{1}{2}|\rmop{grad}\nolimits_x \Phi(t, x, y)|_g^2+t^2 V(x)=0. 
    \]
    We set $t=0$ and obtain
    \begin{equation}
    -\Phi(0, x, y)+\frac{1}{2}|\rmop{grad}\nolimits_x \Phi(0, x, y)|_g^2=0. \label{eq_stat_HJ_free}
    \end{equation}
    We take local coordinates $\phi =(x_1, \ldots, x_n)$. This induces the local coordinates $\phi\times \phi=(x_1, \ldots, x_n, y_1, \ldots, y_n)$ in $M\times M$. We take a differential $\partial_{y_j}$ of both sides of \eqref{eq_stat_HJ_free}. Then we have
    \[
    -\frac{\partial\Phi}{\partial y_j}(0, x, y)+\sum_{k, l=1}^n g^{kl}(x) \frac{\partial \Phi}{\partial x_k}(0, x, y)\frac{\partial^2 \Phi}{\partial x_l \partial y_j}(0, x, y)=0. 
    \]
    We expand $\partial_{y_j}\Phi$ and $\partial_{x_k}\Phi$ as 
    \begin{equation}\label{eq_mvv_expand_y}
        \begin{split}
            \frac{\partial \Phi}{\partial y_j}(0, x, y)
    &=\frac{\partial \Phi}{\partial y_j}(0, y, y)+\sum_{i=1}^n \frac{\partial^2 \Phi}{\partial x_i \partial y_j}(0, y, y)(x_i-y_i)+O(|x-y|^2) \\
    &=\frac{\partial \Phi}{\partial y_j}(0, x, x)+\sum_{i=1}^n \frac{\partial^2 \Phi}{\partial x_i \partial y_j}(0, x, x)(x_i-y_i)+O(|x-y|^2)
        \end{split}
    \end{equation}
    and
    \begin{equation}\label{eq_mvv_expand_x}
    \frac{\partial \Phi}{\partial x_k}(0, x, y)=\frac{\partial \Phi}{\partial x_k}(0, x, x)-\sum_{i=1}^n \frac{\partial^2 \Phi}{\partial x_k \partial y_i}(0, x, x)(x_i-y_i)+O(|x-y|^2). 
    \end{equation}

    Recalling that $\Phi(0, x, y)=d(x, y)^2/2$, we obtain $\partial_{y_j}\Phi(0, x, x)=\partial_{x_k}\Phi(0, x, x)=0$. Thus \eqref{eq_mvv_expand_y} and \eqref{eq_mvv_expand_x} become

    \begin{equation}\label{eq_mvv_expand_y2}
        \begin{split}
            \frac{\partial \Phi}{\partial y_j}(0, x, y)
    &=\sum_{i=1}^n \frac{\partial^2 \Phi}{\partial x_i \partial y_j}(0, y, y)(x_i-y_i)+O(|x-y|^2) \\
    &=\sum_{i=1}^n \frac{\partial^2 \Phi}{\partial x_i \partial y_j}(0, x, x)(x_i-y_i)+O(|x-y|^2)
        \end{split}
    \end{equation}
    and
    \begin{equation}\label{eq_mvv_expand_x2}
        \frac{\partial \Phi}{\partial x_k}(0, x, y)=\frac{\partial \Phi}{\partial x_k}(0, x, x)-\sum_{i=1}^n \frac{\partial^2 \Phi}{\partial x_k \partial y_i}(0, x, x)(x_i-y_i)+O(|x-y|^2)
        \end{equation}
    respectively. 
    We substitute \eqref{eq_mvv_expand_y2} and \eqref{eq_mvv_expand_x2} to \eqref{eq_stat_HJ_free} and obtain 
    \begin{align*}
    &\sum_{i=1}^n \left( -\frac{\partial^2 \Phi}{\partial x_i \partial y_j}(0, x, x)-\sum_{k, l=1}^n g^{kl}(x) \frac{\partial^2 \Phi}{\partial x_k\partial y_i}(0, x, x)\frac{\partial^2 \Phi}{\partial x_l \partial y_j}(0, x, x)\right) (x_i-y_i) \\
    &=O(|x-y|^2). 
    \end{align*}
    Since this holds for all $y$ near $x$, we obtain
    \[
    -\frac{\partial^2 \Phi}{\partial x_i \partial y_j}(0, x, x)+\sum_{k, l=1}^n g^{kl}(x) \frac{\partial^2 \Phi}{\partial x_k\partial y_i}(0, x, x)\frac{\partial^2 \Phi}{\partial x_l \partial y_j}(0, x, x)=0
    \]
    for all indices $i, j$. In the matrix form, this is equivalent to
    \[
    -A(x)-\transp{A(x)}G(x)^{-1}A(x)=0, 
    \]
    where $A(x):=(\partial_{x_i}\partial_{y_j}\Phi(0, x, x))_{i, j}$ and $G(x):=(g_{ij}(x))_{i, j}$. Since $A(x)$ is a regular matrix, we obtain
    \[
    \det (-G(x)^{-1}A(x))=1. 
    \]
    The left hand side is equal to $D_\Phi (0, x, x)$. 
    \end{proof}

Proposition \ref{prop_amplitude_diagonal} shows that, taking $t_0>0$ and a neighborhood $N$ of $\rmop{diag}(M)$ if necessary, we can define the square root $\sqrt{D(t, x, y)}=t^{-n/2}\sqrt{D_\Phi (t, x, y)}$ in $(0, t_0] \times N$. Now we summarize the condition for $t_0>0$ and $N\supset \rmop{diag}(M)$. 

    Throughout the paper, we take a small number $t_0>0$ and a small neighborhood $N$ of $\rmop{diag}(M)$ such that the assertion in Theorem \ref{theo_action_well_defined}, that in Theorem \ref{theo_phase_function} and the positivity of $D_\Phi|_{[0, t_0]\times N}$ hold. 


\subsubsection{Transport equations}
The transport equation (Theorem \ref{theo_1/2_transport}) and the diagonal value of the Laplacian (Theorem \ref{theo_curvature}) are employed in order to prove the consistency.

We begin with a transport equation which the Morette-Van Vleck itself satisfies.  

\begin{theo}\label{theo_1_transport_equation}
The Morette-Van Vleck determinant defined by \eqref{eq_morette_van_vleck} satisfies the transport equation
\begin{equation}\label{eq_1_transport}
\frac{\partial D}{\partial t} (t, x, y)+g(\mathop{\mathrm{grad}}\nolimits_x S(t, x, y), \rmop{grad}\nolimits_x D(t, x, y))+D(t, x, y)\triangle_x S(t, x, y)=0. 
\end{equation}
\end{theo}

\begin{proof}
Let $(x, y)\in N$. We use the same notation in the definition of the Morette-Van Vleck determinant \eqref{eq_morette_van_vleck}. The time derivative of the Morette-Van Vleck determinant is 
\begin{align*}
\partial_t D=(-1)^n g_\lambda (x)^{-1/2}g_\mu(y)^{-1/2} \sum_{j=1}^n \det 
\begin{pmatrix}
\partial_{x_1}\partial_{y_1}S & \cdots & \partial_{x_1}\partial_{y_n}S \\
 \vdots & & \vdots \\
\partial_{x_{j-1}}\partial_{y_1}S & \cdots & \partial_{x_{j-1}}\partial_{y_n}S \\
\partial_{x_j}\partial_{y_1}\partial_t S & \cdots & \partial_{x_j}\partial_{y_n}\partial_t S \\
\partial_{x_{j+1}}\partial_{y_1}S & \cdots & \partial_{x_{j+1}}\partial_{y_n}S \\
 \vdots & & \vdots \\
\partial_{x_n}\partial_{y_1}S & \cdots & \partial_{x_n}\partial_{y_n}S 
\end{pmatrix}. 
\end{align*}
The Hamilton-Jacobi equation (Theorem \ref{theo_property_of_action}) yields 
\begin{align*}
\partial_{x_j}\partial_{y_k}\partial_t S
&=-\partial_{x_j}\partial_{y_k}\left( \frac{1}{2}|\rmop{grad}\nolimits_x S|_g^2+V(x)\right) \\
&=-\sum_{l, m=1}^n \partial_{x_j} (g_\lambda^{lm}(x) \partial_{x_l}\partial_{y_k}S \cdot \partial_{x_m}S) \\
&=-\sum_{l, m=1}^n \partial_{x_j} (g_\lambda (x)^{-1/2} \partial_{x_l}\partial_{y_k}S) \cdot g_\lambda(x)^{1/2} g_\lambda^{lm}(x)\partial_{x_m}S \\
&\quad -\sum_{l, m=1}^n  g_\lambda (x)^{-1/2} \partial_{x_l}\partial_{y_k}S \cdot \partial_{x_j}(g_\lambda(x)^{1/2} g_\lambda^{lm}(x)\partial_{x_m}S). 
\end{align*}
Here $(g^{jk}_\lambda (x))_{j,k}$ is the inverse matrix of $(g^\lambda_{jk})_{j,k}$ defined by $g_x=\sum g^\lambda_{jk}(x)dx_j dx_k$. Thus if we set 
\[
A_{jl}:=
\begin{pmatrix}
\partial_{x_1}\partial_{y_1}S & \cdots & \partial_{x_1}\partial_{y_n}S \\
 \vdots & & \vdots \\
\partial_{x_{j-1}}\partial_{y_1}S & \cdots & \partial_{x_{j-1}}\partial_{y_n}S \\
\partial_{x_j} (g_\lambda (x)^{-1/2} \partial_{x_l}\partial_{y_1}S) & \cdots & \partial_{x_j} (g_\lambda (x)^{-1/2} \partial_{x_l}\partial_{y_n}S) \\
\partial_{x_{j+1}}\partial_{y_1}S & \cdots & \partial_{x_{j+1}}\partial_{y_n}S \\
 \vdots & & \vdots \\
\partial_{x_n}\partial_{y_1}S & \cdots & \partial_{x_n}\partial_{y_n}S 
\end{pmatrix}
\]
and 
\[
B_{jl}:=
\begin{pmatrix}
\partial_{x_1}\partial_{y_1}S & \cdots & \partial_{x_1}\partial_{y_n}S \\
 \vdots & & \vdots \\
\partial_{x_{j-1}}\partial_{y_1}S & \cdots & \partial_{x_{j-1}}\partial_{y_n}S \\
\partial_{x_l}\partial_{y_1}S & \cdots & \partial_{x_l}\partial_{y_n}S \\
\partial_{x_{j+1}}\partial_{y_1}S & \cdots & \partial_{x_{j+1}}\partial_{y_n}S \\
 \vdots & & \vdots \\
\partial_{x_n}\partial_{y_1}S & \cdots & \partial_{x_n}\partial_{y_n}S 
\end{pmatrix}, 
\]
then 
\begin{align*}
&\sum_{j=1}^n\det 
\begin{pmatrix}
\partial_{x_1}\partial_{y_1}S & \cdots & \partial_{x_1}\partial_{y_n}S \\
 \vdots & & \vdots \\
\partial_{x_{j-1}}\partial_{y_1}S & \cdots & \partial_{x_{j-1}}\partial_{y_n}S \\
\partial_{x_j}\partial_{y_1}\partial_t S & \cdots & \partial_{x_j}\partial_{y_n}\partial_t S \\
\partial_{x_{j+1}}\partial_{y_1}S & \cdots & \partial_{x_{j+1}}\partial_{y_n}S \\
 \vdots & & \vdots \\
\partial_{x_n}\partial_{y_1}S & \cdots & \partial_{x_n}\partial_{y_n}S 
\end{pmatrix} \\
=&
-\sum_{l, m=1}^n g_\lambda(x)^{1/2} g_\lambda^{lm}(x)\partial_{x_m}S\sum_{j=1}^n \det A_{jl}
 \\
&-\sum_{j, l, m=1}^n g_\lambda(x)^{-1/2}\partial_{x_j}(g_\lambda (x)^{1/2} g_\lambda^{lm}(x)\partial_{x_m}S)\det B_{jl}. 
\end{align*}

Hence 
\begin{align}
&\partial_t D= \nonumber \\
&-(-1)^n g_\lambda (x)^{-1/2}g_\mu(y)^{-1/2}\sum_{l, m=1}^n g_\lambda(x)^{1/2} g_\lambda^{lm}(x)\partial_{x_m}S\sum_{j=1}^n \det A_{jl} \label{eq_transport_matrix}\\
&-(-1)^n g_\lambda (x)^{-1/2}g_\mu(y)^{-1/2}\sum_{j, l, m=1}^n g_\lambda(x)^{-1/2}\partial_{x_j}(g_\lambda (x)^{1/2} g_\lambda^{lm}(x)\partial_{x_m}S)\det B_{jl}. \nonumber
\end{align}
We claim that 
\begin{equation}\label{eq_Ajl}
\sum_{j=1}^n \det A_{jl}=(-1)^n g_\mu(y)^{1/2}\partial_{x_l}D(t, x, y)
\end{equation}
and 
\begin{equation}\label{eq_Bjl}
\det B_{jl}=(-1)^n \delta_{jl}g_\lambda (x)^{1/2}g_\mu (y)^{1/2}D(t, x, y). 
\end{equation}
Here $\delta_{jl}$ is the Kronecker delta
\[
\delta_{jl}:=
\begin{cases}
1 & \text{if } j=l, \\
0 & \text{if } j\neq l. 
\end{cases}
\]

Firstly we prove \eqref{eq_Bjl}. If $j\neq l$, then the $j$th row of $B_{jl}$ corresponds to the $l$th row of $B_{jl}$. Thus $\det B_{jl}=0$ for $j\neq l$. If $j=l$, then $B_{jl}=B_{jj}$ is just the matrix $(\partial_{x_j}\partial_{y_k}S)_{j, k=1}^n=-g_\lambda (x)^{1/2}g_\mu (y)^{1/2}D$. Hence \eqref{eq_Bjl} is proved. 

Secondly we prove \eqref{eq_Ajl}. By the Leibnitz rule, we obtain
\[
\partial_{x_j} (g_\lambda (x)^{-1/2} \partial_{x_l}\partial_{y_k}S)
=\partial_{x_j}(g_\lambda (x)^{-1/2}) \cdot\partial_{x_l}\partial_{y_k}S
+g_\lambda (x)^{-1/2} \partial_{x_j}\partial_{x_l}\partial_{y_k}S. 
\]
Thus \eqref{eq_Ajl} is proved by the calculation
\begin{align*}
&\sum_{j=1}^n \det A_{jl} \\
&=\sum_{j=1}^n \partial_{x_j}g_\lambda (x)^{-1/2} \cdot \det B_{jl}+g_\lambda (x)^{-1/2} \partial_{x_l} \det (\partial_{x_j}\partial_{y_k}S)_{j, k=1}^n \\
&=\partial_{x_l} (g_\lambda (x)^{-1/2})\cdot \det (\partial_{x_j}\partial_{y_k}S)_{j, k=1}^n+g_\lambda (x)^{-1/2} \partial_{x_l} \det (\partial_{x_j}\partial_{y_k}S)_{j, k=1}^n \\
&=\partial_{x_l}(g_\lambda (x)^{-1/2}\det (\partial_{x_j}\partial_{y_k}S)_{j, k=1}^n) \\
&=(-1)^n g_\mu(y)^{1/2}\partial_{x_l}D(t, x, y). 
\end{align*}

By \eqref{eq_Ajl} and \eqref{eq_Bjl}, we obtain 
\begin{align*}
&\sum_{l, m=1}^n g_\lambda(x)^{1/2} g_\lambda^{lm}(x)\partial_{x_m}S\sum_{j=1}^n \det A_{jl} \\
&=(-1)^ng_\lambda(x)^{1/2}g_\mu(y)^{1/2} \sum_{l, m=1}^n g_\lambda^{lm}(x)\partial_{x_m}S \cdot \partial_{x_l}D(t, x, y) \\
&=(-1)^n g_\lambda(x)^{1/2}g_\mu(y)^{1/2} g(\rmop{grad}\nolimits_x S, \mathrm{grad}_x\, D)
\end{align*}
and 
\begin{align*}
&\sum_{j, l, m=1}^n g_\lambda(x)^{-1/2}\partial_{x_j}(g_\lambda (x)^{1/2} g_\lambda^{lm}(x)\partial_{x_m}S)\det B_{jl} \\
&=(-1)^n g_\lambda(x)^{1/2}g_\mu(y)^{1/2}D\sum_{l, m=1}^n g_\lambda(x)^{-1/2}\partial_{x_l}(g_\lambda (x)^{1/2} g_\lambda^{lm}(x)\partial_{x_m}S) \\
&=(-1)^n g_\lambda(x)^{1/2}g_\mu(y)^{1/2}D\triangle_x S. 
\end{align*}
Substituting these to \eqref{eq_transport_matrix}, we obtain
\[
\partial_t D=-g(\rmop{grad}\nolimits_x S, \mathrm{grad}_x\, D)-D\triangle_x S. \qedhere
\]
\end{proof}

Theorem \ref{theo_1_transport_equation} implies that $\sqrt{D}$ satisfies another transport equation: 
\begin{theo}\label{theo_1/2_transport}
    The square root $\sqrt{D}$ of the Morette-Van Vleck determinant $D$ satisfies the transport equation
\begin{align*}
&\frac{\partial}{\partial t} \sqrt{D(t, x, y)}+g(\rmop{grad}\nolimits_x S(t, x, y), \rmop{grad}\nolimits_x \sqrt{D(t, x, y)}) \\
&+\frac{1}{2}\sqrt{D(t, x, y)}\triangle_x S(t, x, y)=0. 
\end{align*}
\end{theo}


\subsubsection{Laplacian of the Morette-Van Vleck determinant}

Finally we calculate the Laplacian of the square root of the Van Vlack determinant on the diagonal, which is the origin of the modification term $R(x)/12$ of the Hamiltonian \eqref{eq_modified_hamiltonian}. 

\begin{theo}\label{theo_curvature}
    For any point $y\in M$, 
    \[
    \triangle_x \sqrt{D_\Phi (0, x, y)}|_{x=y}=\frac{1}{6}R(y). 
    \]
    \end{theo}
    
    \begin{proof}
    We begin with the transport equation \eqref{eq_1_transport} in Theorem \ref{theo_1_transport_equation}: 
    \[
    \frac{\partial D}{\partial t} (t, x, y)+g(\rmop{grad}\nolimits_x S(t, x, y), D(t, x, y))+D(t, x, y)\triangle_x S(t, x, y)=0. 
    \]
    Substituting $S(t, x, y)=t^{-1}\Phi(t, x, y)$ and $D(t, x, y)=t^{-n}D_\Phi(t, x, y)$, we obtain
    \[
    -nD_\Phi+t\frac{\partial D_\Phi}{\partial t}+g(\rmop{grad}\nolimits_x \Phi, \rmop{grad}\nolimits_x D_\Phi )+D_\Phi\triangle_x \Phi=0. 
    \]
    Now we put $t=0$. Then we have
    \begin{equation}\label{eq_transport_t=0}
        \begin{split}
            &-nD_\Phi(0, x, y)+g(\rmop{grad}\nolimits_x \Phi(0, x, y), \rmop{grad}\nolimits_xD_\Phi(0, x, y)) \\
    &+D_\Phi(0, x, y)\triangle_x \Phi(0, x, y) 
    =0. 
        \end{split}
    \end{equation}
    We recall that $\Phi(0, x, y)=d(x, y)^2/2$. Take a normal coordinate $(x_1, \ldots, x_n)$ with center at $y$. Then 
    \begin{equation}\label{eq_normal_Phi}
        \Phi(0, x, y)=\frac{1}{2}\sum_{j=1}^n x_j^2
    \end{equation}
    if the coordinate function maps $x$ to $(x_1, \ldots, x_n)$. We calculate $\triangle_x \Phi(0, x, y)$. We employ the well-known formula
    \[
    g_{ij}(x)=\delta_{ij}-\frac{1}{3}\sum_{k, l=1}^n R_{ikjl} x_k x_l +O(|x|^3) \quad (|x|\to 0)
    \]
    in the normal coordinates, where $\delta_{ij}$ is the Kronecker delta and $R_{ikjl}$ is the Riemann curvature tensor at the center of the normal coordinates, associated with the Levi-Civita connection $\nabla$ of $g$: 
    \[
    R_{ikjl}:=g\left( \frac{\partial}{\partial x_l}, \left(\nabla_\frac{\partial}{\partial x_j}\nabla_\frac{\partial}{\partial x_l}-\nabla_\frac{\partial}{\partial x_l}\nabla_\frac{\partial}{\partial x_j}\right) \frac{\partial}{\partial x_k}\right). 
    \]
    Thus the inverse matrix $(g^{ij}(x))$ of $(g_{ij}(x))$ and the volume density $\sqrt{g(x)}:=\sqrt{\det (g_{ij}(x))}$ has the expansion
    \[
    g^{ij}(x)=\delta_{ij}+\frac{1}{3}\sum_{k, l=1}^n R_{ikjl}x_k x_l+O(|x|^3)
    \]
    and 
    \begin{align*}
    \sqrt{g(x)}
    &=\exp \left( \frac{1}{2}\rmop{tr} \log (g_{ij}(x))\right)=\exp \left( \frac{1}{2} \rmop{tr}   \left( -\frac{1}{3}\sum_{k, l=1}^n R_{ikjl}x_kx_l +O(|x|^3)\right)\right) \\
    &=\exp \left(-\frac{1}{6} \sum_{k, l=1}^nR_{kl}x_k x_l +O(|x|^3)\right)
    =1-\frac{1}{6} \sum_{k, l=1}^nR_{kl}x_k x_l +O(|x|^3). 
    \end{align*}
    Here $R_{kl}:=\sum_{i=1}^n R_{ikil}$ is the Ricci curvature. Thus 
    \begin{align}
            \triangle_x \Phi(0, x, y)
    &=\frac{1}{\sqrt{g(x)}}\sum_{j, k=1}^n \frac{\partial}{\partial x_j} \left( \sqrt{g(x)}g^{jk}(x)\frac{\partial}{\partial x_k} \left( \frac{1}{2}\sum_{l=1}^n x_l^2\right) \right) \nonumber \\
    &=n-\frac{1}{3}\sum_{k, l=1}^n R_{kl}x_k x_l+O(|x|^3). \label{eq_lap_Phi}
    \end{align}
    By \eqref{eq_normal_Phi} and \eqref{eq_lap_Phi}, the equation \eqref{eq_transport_t=0} becomes
    \begin{equation}
    \sum_{j=1}^n x_j \frac{\partial}{\partial x_j}D_\Phi(0, x, y)-\frac{1}{3}\sum_{k, l=1}^n D_\Phi(0, x, y)R_{kl}x_k x_l =O(|x|^3). \label{eq_transport_asymp}
    \end{equation}
    We expand $D_\Phi(0, x, y)$ as
    \[
    D_\Phi(0, x, y)=1+\sum_{j=1}^n b_j x_j+\sum_{j, k=1}^n c_{jk} x_j x_k+O(|x|^3)
    \]
    employing Proposition \ref{prop_amplitude_diagonal}, that is, $D_\Phi (0, y, y)=1$. 
    Here $b_j, c_{jk}\in \mathbb{R}$ and $c_{jk}=c_{kj}$ for all $j, k$.  
    Thus the equation \eqref{eq_transport_asymp} leads to the relations
    \[
    b_j=0, \quad c_{jk}=\frac{1}{6}R_{jk}. 
    \]
    Hence we obtain 
    \[
    \sqrt{D_\Phi(0, x, y)}=\sqrt{1+\frac{1}{6}\sum_{j, k=1}^n R_{jk}x_j x_k +O(|x|^3)}=1+\frac{1}{12}R_{jk}x_j x_k+O(|x|^3)
    \]
    and 
    \[
    \triangle_x \sqrt{D_\Phi(0, x, y)} |_{x=y}=\left.\left( \frac{R(y)}{6}+O(|x|)\right)\right|_{x_1=\cdots=x_n=0}=\frac{1}{6}R(y)
    \]
    by the definition of the scalar curvature $R(y):=\sum_{j=1}^n R_{jj}$. 
    \end{proof}


\section{Estimates for oscillatory integral operators}\label{sect_FIO}

\subsection{Key estimates}\label{subsec_osc_goal}

In order to prove the stability and the consistency, we need estimates of oscillatory integral operators from Sobolev spaces to the $L^2$ space.


As in the end of Chapter \ref{subsubsec_action}, we take a small number $t_0>0$ and a small neighborhood $N$ of $\rmop{diag}(M)$ such that the assertion in Theorem \ref{theo_action_well_defined}, that in Theorem \ref{theo_phase_function} and the positivity of $D_\Phi|_{[0, t_0]\times N}$ hold. For $a\in C_c^\infty([0, t_0]\times N)$, we define an integral operator $T_\hbar [a](t): C^\infty(M)\to C^\infty(M)$ as

\begin{equation}\label{eq_defi_oscillatory}
    \begin{split}
T_\hbar [a](t)u(x)
&:=\frac{1}{(2\pi i\hbar t)^{n/2}}\int_M a(t, x, y)e^{iS(t, x, y)/\hbar}u(y)\, \vol_g(y) \\
&=\frac{1}{(2\pi i\hbar t)^{n/2}}\int_M a(t, x, y)e^{i\Phi (t, x, y)/\hbar t}u(y)\, \vol_g(y). 
    \end{split}
\end{equation}

Recall the Hamilton flow $\varphi^t_s$ generated by the Hamilton vector field $X_{H_t}$ defined in \eqref{eq_hamiltonian_flow_scaled}. Let $\pi: T^*M\to M$ be the natural projection and we define 
\begin{equation}\label{eq_configuration_flow_scaled}
    q^*(t, y, \eta):=\pi(\varphi^t_1(y, \eta)).
\end{equation} 
This function $q^*(t, y, \eta)$ is necessary to find an upper bound of the operator norm of the operator $T_\hbar [a](t)$ (Theorem \ref{theo_L2_boundedness}). 


We state two key estimates for the proof of the stability and the consistency. The first estimate is an estimate of $L^2$-operator norm of oscillatory integral operators, which plays a crucial role in the proof of stability. 

\begin{theo}\label{theo_L2_boundedness}
Suppose that $a(t, x, y)\in C_c^\infty([0, t_0]\times N)$. Then there exists $C>0$ such that for all $t\in (0, t_0/2]$ and $\hbar \in (0, 1]$, the inequality
\[
\|T_\hbar [a](t)\|_{L^2(M)\to L^2(M)}^4 \leq \left\|\frac{|a(t, q^*(t, y, \eta), y)|^2}{|D_\Phi (t, q^*(t, y, \eta), y)|}\right\|_{L^\infty(T^*M)}^2+C\hbar t
\] 
holds. 
\end{theo}

The second estimate is an estimate of the operator norm on the Sobolev spaces, which is employed for proving the consistency. Roughly speaking, the operator norm of $T_\hbar[a](t)$ is mainly dominated by the behavior of the amplitude near $\rmop{diag}(M)$: 

\begin{theo}\label{theo_vanishing}
Let $J$ and $L$ be nonnegative integers. Assume that $\partial_t^l\partial_x^\alpha\partial_y^\beta a(0, x, y)|_{x=y}=0$ for all $l+|\alpha|+|\beta|\leq 2J+L$. Take a family of real numbers $\{ s_\alpha\in [0, |\alpha|]\}_{|\alpha|\leq L+1}$ and define $s:=\max_{|\alpha|\leq L+1} s_\alpha$. Then there exists a constant $C>0$ such that 
\[
\| T_\hbar[a](t)\|_{H_\hbar^{s/2}\to L^2}^2\leq C\left(\sum_{\substack{|\alpha|+l=L+1 \\ l\geq 0}} t^{l+s_\alpha}+(\hbar t)^{J+1}\right)
\]
holds for all $t\in (0, t_0/2]$ and $\hbar\in (0, 1]$. 
\end{theo}

The goal of this section is to prove Theorem \ref{theo_L2_boundedness} and Theorem \ref{theo_vanishing}. 

Before the proof, we remark that the meaning of $t$ in the phase function $\Phi(t, x, y)/\hbar t$ in \eqref{eq_defi_oscillatory}. The $t\in [-t_0, t_0]$ in the argument of $\Phi(t, x, y)$ is the parameter of the perturbation by the potential $t^2V(x)$ as mentioned in Subsection \ref{subsec_classical_mechanics}. This does not affect the behavior of the oscillatory integral operators so much. However $t\in (0, t_0]$ in the denominator of $\Phi(t, x, y)/\hbar t$ behaves like ``\textit{a semiclassical parameter}.'' For example, the $t$-dependence of the remainder term $C\hbar t$ in Theorem \ref{theo_L2_boundedness} is derived by regarding $\hbar t$ as a semiclassical parameter. Another example is Proposition \ref{prop_fundamental_properties} (ii). We apply the stationary phase method to calculate the limit of the short-time approximate solution $E_\hbar (t)$ as $t\to +0$ (see Subsection \ref{subsect_proof_of_fp} for more details). Thus we treat $t$ in $\Phi(t, x, y)$ and that in the denominator of $\Phi(t, x, y)/\hbar t$ as if they are independent variables (except for Theorem \ref{theo_vanishing}). 

\subsection{Semiclassical pseudodifferential operators on manifolds}
We employ a theory of pseudodifferential operators on manifolds. For more details, see Appendix E of \cite{Dyatlov-Zworski19}. 

We begin with a definition of semiclassical Weyl quantization on Euclidean spaces. 

\begin{defi}
We define symbol classes $S^m_{1, 0}$ for $m\in \mathbb{R}$ as 
\[
S^m_{1, 0}:=\left\{\, a(x, \xi)\in C^\infty(\mathbb{R}^n\times \mathbb{R}^n) \,\middle|\, 
\begin{aligned}
&\|(1+|\xi|)^{-m+|\beta|}\partial_x^\alpha\partial_\xi^\beta a(x, \xi)\|_{L^\infty}<\infty \\
&\text{ for all } \alpha, \beta \in \mathbb{N}_0^n
\end{aligned}
\,\right\}. 
\]
For any symbol $a\in S^m_{1, 0}$, we define a semiclassical Weyl quantization $a^\mathrm{w}(x, \hbar D)$ of $a$ as 
\[
a^\mathrm{w}(x, \hbar D)u(x):=\frac{1}{(2\pi \hbar)^n}\int_{\mathbb{R}^n} d\xi \int_{\mathbb{R}^n} dy\, a\left(\frac{x+y}{2}, \xi\right) e^{i\xi\cdot (x-y)/\hbar} u(y) \quad (u\in \mathscr{S}), 
\]
where $\mathscr{S}$ is the space of Schwartz class (rapidly decreasing) functions. 
\end{defi}

Next we consider pseudodifferential operators on manifolds. Let $M$ be a manifold with fixed smooth density. 

\begin{defi}\label{defi_psiDO}
A family $\{A_\hbar\}_{\hbar\in (0, 1]}$ of operators acting on functions on $M$ is an $\hbar$-pseudodifferential operator of class $\Psi_\hbar^m$ on $M$ if $A_\hbar$ is represented as 
\begin{equation}
A_\hbar=\sum_{j=1}^N \chi_j \psi_j^* a_j^\mathrm{w}(\hbar; x, \hbar D) \psi_j^{-1*} \chi_j + B_\hbar, \label{eq_psiDO_manifold}
\end{equation}
where 
\begin{itemize}
\item $\chi_j$ are multiplication operators by smooth functions $\chi_j\in C_c^\infty (U_\lambda)$ supported in some coordinate neighborhood $U_\lambda=U_{\lambda_j}$; 
\item $\psi_j: U_\lambda \to \mathbb{R}^n$ is a coordinate function; 
\item $a_j(\hbar; \cdot, \cdot)\in S^m_{1, 0}$ has an asymptotic expansion
\[
a_j(\hbar; x, \xi)\sim \sum_{k=0}^\infty \hbar^k a_{jk}(x, \xi), \quad a_{jk}\in S^{m-k}_{1, 0}
\]
and $\{ a(\hbar; \cdot, \cdot)\}_{\hbar\in (0, 1]}$ forms a bounded family in $S^m_{1, 0}$;  
\item $B_\hbar$ is an operator with smooth kernel $B_\hbar(x, y)$ such that 
\begin{equation}
\| B_\hbar \|_{C^s (K)} \leq C_{sNK}\hbar^N \label{eq_smoothing_op}
\end{equation}
for all compact sets $K\subset M\times M$, $\hbar\in (0, 1]$ and nonnegative integers $s$ and $N$. 

We denote a class consisting of operators $B_\hbar$ satisfying \eqref{eq_smoothing_op} by $\hbar^\infty \Psi_\hbar^{-\infty}$. 
\end{itemize}
\end{defi}

If a pseudodifferential operator $A_\hbar\in \Psi_\hbar^m$ has a representation \eqref{eq_psiDO_manifold}, the function
\begin{equation}\label{eq_defi_principal_symbol}
\sigma_\hbar (A_\hbar)(x, \xi):=\sum_{j=1}^N \chi_j(x)^2 a_{j0}(\tilde\psi_j(x, \xi))
\end{equation}
is independent of the representation of $A_\hbar$. Here $\tilde \psi_j: T^*U_\lambda \to T^*\mathbb{R}^n$ is the canonical coordinates associated with $\psi_j: U_\lambda\to \mathbb{R}^n$. This $\sigma_\hbar(A_\hbar)$ is called the \textit{principal symbol} of $A_\hbar$.  


Then the following facts hold. The proof is in $\cite{Dyatlov-Zworski19}$. 

\begin{theo}\label{theo_psiDO_basic}
\begin{enumerate}
\item If $A_\hbar \in \Psi_\hbar^k$ and $B_\hbar\in \Psi_\hbar^m$, then $A_\hbar B_\hbar \in \Psi_\hbar^{k+m}$ with principal symbol $\sigma_\hbar(A_\hbar)\sigma_\hbar (B_\hbar)$. 
\item If $A_\hbar\in \Psi_\hbar^k$, then $A_\hbar^*\in \Psi_\hbar^k$ and $\sigma_\hbar (A_\hbar^*)=\overline{\sigma_\hbar (A_\hbar)}$. 
\item (sharp G\aa rding inequality) If (the Schwartz kernel of) $A_\hbar\in \Psi^0_\hbar$ has a compact support and $\sigma (A_\hbar)\geq 0$, then there exists a positive constant $C>0$ such that 
\[
\jbracket{A_\hbar u, u}_{L^2} \geq -C\hbar \| u\|_{L^2}^2
\]
for all compactly supported $u\in L^2(M)$ and $\hbar\in (0, 1]$. 
\end{enumerate}
\end{theo}

In our setting, we can apply the theory of pseudodifferential operators by considering $T_\hbar[a](t)^* T_\hbar [a](t)$, regarding $\hbar t$ as a semiclassical parameter independent of $t$ in $\Phi(t, x, y)$ and $q^*(t, x, \eta)$: 

\begin{theo}\label{theo_psiDO}
The operator $T_\hbar[a](t)^*T_\hbar[a](t): C_c^\infty(M)\to C_c^\infty(M)$ is a $(\hbar t)$-pseudodifferential operator of class $\Psi^0_{\hbar t}$ with the principal symbol 
\[
|a(t, q^*(t, x, \eta), y)|^2|D_\Phi (t, q^*(t, x, \eta), y)|^{-1} 
\]
with the parameters $t\in (0, t_0/2]$ and $\hbar\in (0, 1]$. 
\end{theo}

We prove Theorem \ref{theo_L2_boundedness} from Theorem \ref{theo_psiDO}. We prepare for the proof of Theorem \ref{theo_psiDO} in Subsection \ref{subsec_disjoint_support} and Subsection \ref{subsec_covering} below. We prove Theorem \ref{theo_psiDO} in Subsection \ref{subsec_proof_FIO}.

\subsection{Disjoint support}\label{subsec_disjoint_support}


A composition of two oscillatory integral operators with amplitudes with disjoint supports is negligible. 

\begin{theo}\label{theo_disjoint_support}
Let $K_1$ and $K_2$ be disjoint compact subsets of $N$. Then $T_\hbar [a](t)^* T_\hbar [b](t)$ is in the class $(\hbar t)^\infty \Psi_{\hbar t}^{-\infty}$. More precisely, for all $J\in\mathbb{N}_0$ and $s\in \mathbb{N}_0$, there exists a constant $C>0$ such that for all $a, b\in C^\infty([0, t_0]\times N)$ with $\rmop{supp}a\subset [0, t_0]\times K_1$ and $\rmop{supp}b\subset [0, t_0]\times K_2$, the integral kernel $K_\hbar (t, y, z)$ of $T_\hbar [a](t)^*T_\hbar [b](t)$ has the estimate 
\[
\|K_\hbar (t, \cdot, \cdot)\|_{C^s(M\times M)}\leq C(\hbar t)^J\|a(t, \cdot, \cdot)\|_{C^{2s+J+n}(K_1)}\|b(t, \cdot, \cdot)\|_{C^{2s+J+n}(K_2)}
\]
for all $t\in (0, t_0]$ and $\hbar\in (0, 1]$. 
\end{theo}

\begin{proof}
We denote the integral kernel of $T_\hbar [a](t)^*T_\hbar [b](t)$ by $K_\hbar (t, y, z)$: 
\begin{equation}
K_\hbar (t, y, z)=\frac{1}{(2\pi\hbar t)^n}\int_M \overline{a(t, x, y)}b(t, x, z)e^{i(\Phi (t, x, z)-\Phi(t, x, y))/\hbar t}\, \vol_g(x). \label{eq_disjoint_kernel}
\end{equation}
This is a smooth function with respect to $(t, y, z)$. On the support of the integrand in \eqref{eq_disjoint_kernel}, the vector field $x\mapsto \mathrm{grad}_x\Phi (t, x, z)-\mathrm{grad}_x\Phi (t, x, y)\in T_xM$ does not vanish since the assertion in Theorem \ref{theo_phase_function} implies the injectivity of $y\mapsto d_x\Phi(t, x, y)$. Thus we can consider the differential operator
\[
Lf(t, x, y, z):=\frac{g_x(\mathrm{grad}_x\Phi (t, x, z)-\mathrm{grad}_x\Phi (t, x, y), \mathrm{grad}_x f(x))}{i|\mathrm{grad}_x \Phi (t, x, y)-\mathrm{grad}_x\Phi (t, x, z)|_g^2}, 
\]
which satisfies
\[
\hbar t Le^{i(\Phi (t, x, z)-\Phi (t, x, y))/\hbar t}= e^{i(\Phi (t, x, z)-\Phi (t, x, y))/\hbar t}. 
\]
By the Gauss divergence theorem, the transpose of $L$ is
\[
\transp{L}f(t, x, y, z)=i\,\mathrm{div}_x \left(\frac{f(x)(\mathrm{grad}_x\Phi (t, x, z)-\mathrm{grad}_x\Phi (t, x, y))}{|\mathrm{grad}_x \Phi (t, x, y)-\mathrm{grad}_x\Phi (t, x, z)|_g^2}\right). 
\]
Hence \eqref{eq_disjoint_kernel} is equivalent to 
\[
K_\hbar (t, y, z)=\frac{(\hbar t)^J}{(2\pi\hbar t)^n}\int_M (\transp{L})^J(\overline{a(t, x, y)}b(t, x, z))e^{i(\Phi (t, x, z)-\Phi (t, x, y))/\hbar t}\, \vol_g(x)
\]
for all $J\in\mathbb{N}_0$. Thus if $A: C^\infty(M\times M)\to C^\infty(M\times M)$ is a differential operator of degree $s\in \mathbb{N}_0$, then 
\[
|AK_\hbar (t, y, z)|\leq C(\hbar t)^{J-n-s}\|a(t, \cdot, \cdot)\|_{C^{s+J}(K_1)}\|b(t, \cdot, \cdot)\|_{C^{s+J}(K_2)}
\]
for some $C>0$ independent of $(t, y, z)\in (0, t_0/2]\times M\times M$, $a\in C^\infty([0, t_0]\times N)$ with $\rmop{supp}a\subset [0, t_0]\times K_1$, $b\in C^\infty([0, t_0]\times N)$ with $\rmop{supp}b\subset K_2$ and $\hbar\in (0, 1]$. 
\end{proof}

\subsection{Covering}\label{subsec_covering}

Let $\{\, \phi_\lambda: U_\lambda\to V_\lambda\,\}_{\lambda\in\Lambda}$ be a local coordinate system on $M$. Then it induces a local coordinate system $\{\phi_\lambda\times \phi_\mu: U_\lambda\times U_\mu \to V_\lambda\times V_\mu\,\}_{(\lambda, \mu)\in \Lambda\times \Lambda}$ on $M\times M$. 
For $(\lambda, \mu)\in \Lambda\times \Lambda$, we define 
\[
\Phi_{\lambda \mu }(t, x, y):=\Phi(t, \phi_{\lambda }^{-1}(x), \phi_{\mu }^{-1}(y)). 
\]

We prepare a lemma for regarding a principal part of $T_\hbar [a](t)^*T_\hbar [a](t)$ as a pseudodifferential operator. 

\begin{lemm}\label{lemm_covering_B}
There exists a finite collection of open sets $\{\Omega_\iota\subset [-t_0, t_0]\times N\}_{\iota\in I}$ which satisfies the following properties. 
\begin{enumerate}
\item $[-t_0/2, t_0/2]\times \rmop{diag}(M)\subset \bigcup_{\iota\in I} \Omega_\iota$. 
\item If $\Omega_\iota\cap \Omega_{\iota^\prime}\neq \varnothing$, then there exist $(\lambda, \mu)\in \Lambda\times \Lambda$, open sets $I^\prime\subset [-t_0, t_0]$, $V\subset V_{\lambda}$ and an open convex subset $W\subset V_{\mu}$ such that $\Omega_\iota \cup \Omega_{\iota^\prime} \subset I\times U_{\lambda } \times U_{\mu }$, $\Omega_\iota\cup \Omega_{\iota^\prime}\subset I^\prime \times \phi_{\lambda }^{-1}(V) \times \phi_{\mu }^{-1}(W)$ and the function 
\[
F: I^\prime \times V \times W \times W \longrightarrow I^\prime\times W \times \mathbb{R}^n \times W, 
\]
\[
F(t, x, y, z):=\left( t, y, -\int_0^1 \partial_y\Phi_{\lambda \mu }(t, x, sz+(1-s)y)\, ds, z\right)
\]
is an embedding. 
\end{enumerate}
\end{lemm}

\begin{proof}
We introduce a collection of $\mathscr{B}_{\lambda\mu}$ of open subsets of $[-t_0, t_0]\times N$ as follows. An open subset $\Omega^\prime$ in  $[-t_0, t_0]\times N$ belongs to $\mathscr{B}_{\lambda\mu}$ if and only if the following conditions hold. 
\begin{enumerate}
\renewcommand{\labelenumi}{(\roman{enumi})}
\item $\Omega^\prime$ is a direct product $\Omega^\prime=I^\prime\times V \times W$ of an open interval $I^\prime\subset [-t_0, t_0]$, open sets $V\subset V_\lambda$ and $W\subset V_\mu$. Moreover $W$ is convex and $\phi_\lambda^{-1}(V) \times \phi_\mu^{-1}(W)\subset N$. 
\item The function 
\[
F: I^\prime \times V \times W \times W \longrightarrow I^\prime\times W \times \mathbb{R}^n \times W, 
\]
\[
F(t, x, y, z):=\left( t, y, -\int_0^1 \partial_y\Phi_{\lambda\mu}(t, x, sz+(1-s)y)\, ds, z\right)
\]
is an embedding. 
\end{enumerate}
We define 
\[
\mathscr{B}:=\bigcup_{(\lambda, \mu)\in\Lambda\times \Lambda}\{\, I^\prime\times \phi_\lambda^{-1}(V)\times \phi_\mu^{-1}(W) \mid I^\prime \times V \times W \in \mathscr{B}_{\lambda\mu}\,\}. 
\]
Since 
\[
\det dF(t, x, y, z)|_{y=z}=\det \partial_x\partial_y\Phi_{\lambda\mu}(t, y, y)\neq 0 
\]
for all $(\lambda, \mu)\in \Lambda\times \Lambda$ and $(t, x, y)\in [-t_0/2, t_0/2]\times V_\lambda \times V_\mu$, the inverse function theorem implies that $\mathscr{B}$ forms an open basis of the topology in $[-t_0/2, t_0/2]\times M \times M$. 

We define two projections $\mathrm{pr}_1, \mathrm{pr}_2: M\times M\to M$ as $\mathrm{pr}_1(x, y)=x$ and $\mathrm{pr}_2(x, y)=y$. Then $\tilde g=\mathrm{pr}_1^*g+\mathrm{pr}_2^*h$ defines a Riemannian metric on $M\times M$. This metric induces a distance function $\tilde d$ on $M\times M$. 
We can take a Lebesgue number $\delta>0$ associated with the open covering $\{ ([-t_0/2, t_0/2]\times \rmop{diag}(M))\cap \Omega \}_{\Omega\in \mathscr{B}}$ on the compact metric space $([-t_0/2, t_0/2]\times \rmop{diag}(M), \rho)$, $\rho((t, p), (s, q)):=|t-s|+\tilde d(p, q)$. 

Now we define 
\[\mathscr{A}:=\left\{\, \Omega\cap \left(\left[-\frac{t_0}{2}, \frac{t_0}{2}\right]\times \rmop{diag}(M)\right)\,\middle|\, \Omega\in \mathscr{B}, \, \rmop{diam} \Omega <\frac{\delta}{2}\, \right\}. 
    \]
$\mathscr{A}$ covers a compact set $[-t_0/2, t_0/2]\times \rmop{diag}(M)$ since $\mathscr{B}$ is a basis of the topology in $[-t_0, t_0]\times M\times M$. Thus we can choose a finite collection $\{(\Omega_\iota\cap ([-t_0/2, t_0/2]\times \rmop{diag}(M)) \in \mathscr{A}\}_{\iota\in I}$ such that $\Omega_\iota\in \mathscr{B}$, $\rmop{diam}\Omega_\iota <\delta/2$ and 
\[\bigcup_{\iota\in I} \left(\Omega_\iota\cap \left(\left[-\frac{t_0}{2}, \frac{t_0}{2}\right]\times \rmop{diag}(M)\right)\right)=\left[-\frac{t_0}{2}, \frac{t_0}{2}\right]\times \rmop{diag}(M). \]

Suppose $\Omega_\iota \cap \Omega_{\iota^\prime} \neq \varnothing$. Then we can prove 
\begin{equation}\mathrm{diam}(\Omega_\iota \cup \Omega_{\iota^\prime})\leq \rmop{diam}\Omega_\iota +\rmop{diam}\Omega_{\iota^\prime}\label{eq_triangle_diam}
\end{equation} as follows. Take arbitrary two points $x, y\in \Omega_\iota \cup \Omega_{\iota^\prime}$. If $x, y\in \Omega_\iota$ or $x, y\in \Omega_{\iota^\prime}$, then we obviously have 
\[
    \rho(x, y)\leq \rmop{diam}\Omega_\iota +\rmop{diam}\Omega_{\iota^\prime}. \]
Without loss of generality, we can suppose that $x\in \Omega_\iota$ and $y\in \Omega_{\iota^\prime}$. Since $\Omega_\iota \cap \Omega_{\iota^\prime}\neq \varnothing$, we can take a point $z\in \Omega_\iota \cap \Omega_{\iota^\prime}$. Then by the triangle inequality, we have 
\[
    \rho(x, y)\leq \rho(x, z)+\rho(y, z)\leq \rmop{diam}(\Omega_\iota)+\rmop{diam}(\Omega_{\iota^\prime}). \]
    There we obtain $\rho(x, y)\leq \rmop{diam}\Omega_\iota +\rmop{diam}\Omega_{\iota^\prime}$ for arbitrary two points $x, y\in \Omega_\iota \cup \Omega_{\iota^\prime}$. Thus the inequality \eqref{eq_triangle_diam} is proved. 

Now we apply \eqref{eq_triangle_diam} and obtian 
\[
\mathrm{diam}(\Omega_\iota \cup \Omega_{\iota^\prime})\leq \rmop{diam}\Omega_\iota +\rmop{diam}\Omega_{\iota^\prime}< \frac{\delta}{2}+\frac{\delta}{2}=\delta.
\]
Thus there exists $\Omega\in \mathscr{B}$ such that $\Omega_\iota\cup\Omega_{\iota^\prime} \subset \Omega$ by the definition of the Lebesgue number $\delta$. We choose $(\lambda , \mu )\in\Lambda\times \Lambda$ such that $\Omega=(\mathrm{id}\times \phi_{\lambda }\times \phi_{\mu })^{-1}(\Omega^\prime)$ for some $\Omega^\prime\in \mathscr{B}_{\lambda \mu }$. Then the restriction of $F$ to open subset $\Omega_\iota\cup\Omega_{\iota^\prime}$ is an embedding. 
\end{proof}

Employing Lemma \ref{lemm_covering_B}, we can calculate a principal part of $T_\hbar [a](t)^*T_\hbar [a](t)$. 

\begin{theo}\label{theo_nonempty_intersection}
Then there exists a finite collection of open sets $\{\Omega_\iota\subset [-t_0, t_0]\times N\}_{\iota\in I}$, a family of smooth functions $\{\kappa_\iota\in C_c^\infty(\Omega_\iota; [0, 1])\}_{\iota\in I}$ and $\chi_{\iota\iota^\prime}\in C_c^\infty (U_\mu)$ which satisfy the following properties. 
\begin{enumerate}
\item $[-t_0/2, t_0/2]\times \rmop{diag}(M)\subset \bigcup_{\iota\in I} \Omega_\iota$ and $\sum_\iota \kappa_\iota=1$ on $[-t_0/2, t_0/2]\times \rmop{diag}(M)$. 
\item If $\Omega_\iota\cap \Omega_{\iota^\prime}\neq \varnothing$, then there exist $(\lambda, \mu)\in \Lambda\times \Lambda$, open sets $I^\prime\subset [-t_0, t_0]$, $V\subset V_\lambda$ and an open subset $W\subset V_\mu$ such that $\Omega_\iota\cup \Omega_{\iota^\prime}\subset [-t_0, t_0]\times U_\lambda\times U_\mu$, $\Omega_\iota\cup \Omega_{\iota^\prime}\subset I^\prime \times \phi_\lambda^{-1}(V) \times \phi_\mu^{-1}(W)$ and
\[
T_\hbar [\kappa_\iota a](t)^*T_\hbar [\kappa_{\iota^\prime} a](t)=\chi_{\iota\iota^\prime}\phi_\mu^*b_{\iota\iota^\prime}^\mathrm{w}(\hbar t; t, y, \hbar t D_y)\phi_\mu^{-1*}\chi_{\iota\iota^\prime}
\]
for some $b_{\iota\iota^\prime}\in S^0_{1, 0}$. Moreover $b_{\iota\iota^\prime}(\hbar t; t, y, \eta)$ has an asymptotic expansion 
\begin{equation}
b_{\iota\iota^\prime}(\hbar t; t, y, \eta)\sim \sum_{j=0}^\infty (\hbar t)^j b_{\iota\iota^\prime, j}(t, y, \eta)\quad \text{in } S^0_{1, 0} \label{eq_asymptotic_b}
\end{equation}
and the principal part $b_{\iota\iota^\prime, 0}$ is
\begin{equation}\label{eq_principal_symbol_locally}
b_{\iota\iota^\prime, 0}(t, y, \eta)= \tilde\phi_\mu^{-1*}\left(\frac{(\kappa_\iota\kappa_{\iota^\prime})(q^*(t, y, \eta), y)|a(t, q^*(t, y, \eta), y)|^2}{|D_\Phi (t, q^*(t, x, \eta), y)|}\right). 
\end{equation}
Furthermore, if $\partial_t^l\partial_x^\alpha\partial_y^\beta a(0, x, x)=0$ for all $x\in M$ and $l+|\alpha|+|\beta|\leq 2J+L$, then $\partial_t^l\partial_\eta^\alpha b_{\iota\iota^\prime, j}(t, y, \eta)|_{t=0, \eta=0}=0$ for all $j\leq J$, $|\alpha|+l\leq L$ and $y\in \mathbb{R}^n$. 
\item $\chi_{\iota\iota^\prime}(y)=1$ if there exists $x\in M$ such that $(x, y)\in \rmop{supp} \kappa_\iota \cup \rmop{supp}\kappa_{\iota^\prime}$.  
\end{enumerate}
\end{theo}

\begin{proof}
We take a finite collection of open sets $\{\Omega_\iota\subset [-t_0, t_0]\times N\}_{\iota\in I}$ as in Lemma \ref{lemm_covering_B}. 

Suppose $\Omega_\iota \cap \Omega_{\iota^\prime}\neq \varnothing$. Then we can take $(\lambda, \mu)\in \Lambda\times \Lambda$, open sets $I^\prime\subset [-t_0, t_0]$, $V\subset V_{\lambda }$ and an open convex subset $W\subset V_{\mu }$ such that $\Omega_\iota \cup \Omega_{\iota^\prime} \subset [-t_0, t_0]\times U_{\lambda } \times V_{\mu }$, $\Omega_\iota\cup \Omega_{\iota^\prime}\subset I^\prime \times \phi_\lambda^{-1}(V) \times \phi_\mu^{-1}(W)$ and the function 
\[
F: I^\prime \times V \times W \times W \longrightarrow I^\prime\times W \times \mathbb{R}^n \times W, 
\]
\[
F(t, x, y, z):=\left( t, y, -\int_0^1 \partial_y\Phi_{\lambda \mu }(t, x, sz+(1-s)y)\, ds, z\right)
\]
is an embedding. 

We denote the integral kernel of $\phi_\mu^{-1*}T_\hbar [\kappa_\iota a](t)^*T_\hbar [\kappa_{\iota^\prime}a](t)\phi_\mu^*$ by $K^\prime_\hbar(t, y, z)$. We define $a_\iota(t, x, y):=\kappa_\iota(\phi_\lambda^{-1}(x), \phi_\mu^{-1}(y))a(t, \phi_\lambda^{-1}(x), \phi_\mu^{-1}(y))$ for $\iota\in I$, $g_\lambda (x)^{1/2}dx=\phi_\lambda^{-1*}\vol_g(x)$ and $g_\mu (y)^{1/2}dy=\phi_\mu^{-1*}\vol_g(y)$. Then 
\begin{align*}
&K^\prime_\hbar (t, y, z) \\
&= 
\frac{1}{(2\pi \hbar t)^n}
\int_{\mathbb{R}^n} \overline{a_\iota(t, x, y)}a_{\iota^\prime}(t, x, z)e^{i(\Phi_{\lambda\mu} (t, x, z)-\Phi_{\lambda\mu} (t, x, y))/\hbar t}g_\lambda (x)^{1/2} g_\mu (z)^{1/2}\, dx. 
\end{align*}
The phase function is 
\[
\Phi_{\lambda\mu} (t, x, z)-\Phi_{\lambda\mu}(t, x, y)
=-(y-z)\cdot \int_0^1 \partial_y\Phi_{\lambda\mu} (t, x, sz+(1-s)y)\, ds. 
\]
Since $\Omega_\iota\cap \Omega_{\iota^\prime}\neq \varnothing$, we can change variables 
\begin{equation}\label{eq_changing_variables}
\eta=\eta(t, x, y, z)=-\int_0^1 \partial_y\Phi_{\lambda\mu} (t, x, sz+(1-s)y)\, ds \Longleftrightarrow 
x=\hat x(t, y, \eta, z)
\end{equation}
by Lemma \ref{lemm_covering_B}
and obtain
\begin{equation}\label{eq_bisymbol_psiDO}
    \begin{split}
        &K^\prime_\hbar (t, y, z) \\
&=\frac{1}{(2\pi \hbar t)^n}
\int_{\mathbb{R}^n} \overline{a_\iota(t, x, y)}a_{\iota^\prime}(t, x, z)e^{i\eta\cdot (y-z)/\hbar t}g_\lambda (x)^{1/2} g_\mu (z)^{1/2}\left|\det \frac{\partial \eta}{\partial x}\right|^{-1}\, d\eta. 
    \end{split}
\end{equation}
This is an integral kernel of the $(\hbar t)$-pseudodifferential operator with symbol
\begin{equation}\label{eq_bisymbol}
\tilde b_{\iota\iota^\prime}(t, y, \eta, z):=\overline{a_\iota(t, \hat x, y)}a_{\iota^\prime}(t, \hat x, z)g_\lambda (\hat x)^{1/2} g_\mu (z)^{1/2}\left|\det \frac{\partial \eta}{\partial x}(t ,\hat x, y, z)\right|^{-1}. 
\end{equation}
In other words, \eqref{eq_bisymbol_psiDO} is equivalent to 
\[
    K^\prime_\hbar (t, y, z)=\frac{1}{(2\pi \hbar t)^n} \int_{\mathbb{R}^n} \tilde b_{\iota\iota^\prime}(t, y, \eta, z) e^{i\eta \cdot (y-z)/\hbar t}\, d\eta. \]
Hence the symbol $b_{\iota\iota^\prime}(\hbar t; t, \cdot, \cdot)=O_{S^0_{1, 0}}(1)$ $(t, \hbar \to 0)$ given by 
\[
b_{\iota\iota^\prime}(\hbar t; t, y, \eta):=e^{-i\hbar t\jbracket{\partial_z, \partial_\eta}}\tilde b_{\iota\iota^\prime}\left.\left(t, y-\frac{z}{2}, \eta, y+\frac{z}{2}\right)\right|_{z=0} 
\]
satisfies
\begin{equation}\label{eq_weyl_psiDO_ht}
\int_{\mathbb{R}^n} K^\prime_\hbar (t, y, z)u(z)\, dz=b_{\iota\iota^\prime}^\mathrm{w}(\hbar t; t, y, \hbar t D_y)u(y)
\end{equation}
for all $u\in C^\infty (M)$. (See Theorem 4.20 in \cite{Zworski12}.) 
This symbol $b_{\iota\iota^\prime}(\hbar t; t, y, \eta)$ has an asymptotic expansion 
\begin{equation}\label{eq_expansion_symbol}
b_{\iota\iota^\prime}(\hbar t; t, y, \eta)=\sum_{j=0}^N (\hbar t)^j b_{\iota\iota^\prime, j}(t, y, \eta)+O_{S^0_{1, 0}}((\hbar t)^{N+1}), 
\end{equation}
where 
\begin{equation}\label{eq_expansion_symbol_term}
b_{\iota\iota^\prime, j}(t, y, \eta):=\frac{(-i)^j}{j!} \jbracket{\partial_z, \partial_\eta}^j \tilde b_{\iota\iota^\prime}\left.\left(t, y-\frac{z}{2}, \eta, y+\frac{z}{2}\right)\right|_{z=0}. 
\end{equation}

We claim that $\partial_t^l\partial_x^\alpha\partial_y^\beta a(0, x, x)=0$ for all $x\in M$ and $l+|\alpha|+|\beta|\leq 2J+L$ implies $\partial_t^l\partial_\eta^\alpha b_{\iota\iota^\prime, j}(t, y, \eta)|_{t=0, \eta=0}=0$ for all $j\leq J$ and $|\alpha|+l\leq L$. Firstly, we have 
\begin{equation}\label{eq_diagonal_zero_energy}\hat x(0, y, 0, y)=y\end{equation} by 
\[
\eta(0, y, y, y)=-\int_0^1 \partial_y \Phi_{\lambda\mu}(t, y, y)\, ds=0
\]
and the definition \eqref{eq_changing_variables} of $\hat x(0, y, 0, y)$. Secondly, by \eqref{eq_bisymbol} and \eqref{eq_diagonal_zero_energy}, $\partial_t^l\partial_x^\alpha\partial_y^\beta a(0, x, x)=0$ for all $x\in M$ and $l+|\alpha|+|\beta|\leq 2J+L$ implies $\partial_t^l \partial _y^\alpha \partial_\eta^\beta \partial_z^\gamma \tilde b_{\iota\iota^\prime}(t, y, \eta, z)|_{t=0, y=z, \eta=0}=0$ for all $l+|\alpha|+|\beta|+|\gamma|\leq 2J+L$. Finally, by the fomula \eqref{eq_expansion_symbol_term} of $b_{\iota\iota^\prime, j}(t, y, \eta)$, $\partial_t^l \partial _y^\alpha \partial_\eta^\beta \partial_z^\gamma \tilde b_{\iota\iota^\prime}(t, y, \eta, z)|_{t=0, y=z, \eta=0}=0$ for all $l+|\alpha|+|\beta|+|\gamma|\leq 2J+L$ implies that $\partial_t^l \partial_\eta^\alpha b_{\iota\iota^\prime, j}(t, y, \eta)|_{t=0, \eta=0}=0$ for all $j\leq J$ and $|\alpha|+l\leq L$. Thus the claim is proved. 

In particular, we set $j=0$ in \eqref{eq_expansion_symbol_term} and obtain 
\begin{equation}\label{eq_principal_symbol}
    \begin{split}
b_{\iota\iota^\prime, 0}(t, y, \eta)
&=\tilde b_{\iota\iota^\prime}(t, y, \eta, y) \\
&=\frac{( \overline{a_\iota}a_{\iota^\prime})(t, \hat x(t, y, \eta, y), y)g_\lambda (\hat x(t, y, \eta, y))^{1/2} g_\mu (z)^{1/2}}{\left|\det \partial_x \eta (t ,\hat x(t, y, \eta, y), y, y)\right|}. 
    \end{split}
\end{equation}
Since $\hat x(t, y, \eta, y)=x$ if and only if $\eta=-\partial_y\Phi_{\lambda\mu} (t, x, y)$ by the definition \eqref{eq_changing_variables} of $\hat x(t, y, \eta, z)$, we have 
\[\hat x(t, y, \eta, y)=\phi_\lambda (q^*(t, \tilde\phi_\mu^{-1} (y, \eta))). \]
Here $\tilde \phi_\mu: T^*U_\mu\to V_\mu\times \mathbb{R}^n$ is the canonical coordinates in $T^*M$ associated with the local coordinates $\phi_\mu: U_\mu \to V_\mu$. Moreover 
\begin{equation}\label{eq_derivative_of_eta}
\frac{\partial \eta_j}{\partial x_k}(t, \phi_\lambda (q^*(t, \tilde\phi_\mu^{-1}(y, \eta))), y, y)=-\frac{\partial^2 \Phi_{\lambda\mu} }{\partial x_k\partial y_j}(t, \phi_\lambda (q^*(t, \tilde\phi_\mu^{-1}(y, \eta))), y)
\end{equation}
by differentiating both sides of $\eta=\partial_y\Phi_{\lambda\mu}(t, x, y)$ by $x$ and substituting $x=\hat x(t, y, \eta, y)=\phi_\lambda (q^*(t, \tilde\phi_\mu^{-1}(y, \eta)))$. 
Hence by \eqref{eq_derivative_of_eta} and the definition \eqref{eq_definiton_mvv_scaled}, we obtain 
\begin{align*}
&g_\lambda (\hat x(t, y, \eta, y))^{1/2} g_\mu (y)^{1/2}\left|\det \frac{\partial \eta}{\partial x}(\hat x(t, y, \eta, y))\right|^{-1}  \\
&=|D_\Phi (t, q^*(t, \tilde\phi_\lambda^{-1}(y, \eta)), \phi_\mu^{-1}(y))|^{-1}. 
\end{align*}
Therefore, recallin \eqref{eq_principal_symbol}, we obtain a formula of the principal symbol $b_{\iota\iota^\prime, 0}(t, y, \eta)$ of the $(\hbar t)$-pseudodifferential operator $\phi_\mu^{-1*}T_\hbar [\kappa_\iota a](t)^*T_\hbar [\kappa_{\iota^\prime}a](t)\phi_\mu^*$: 
\[
b_{\iota\iota^\prime, 0}(t, y,\eta)=\tilde\phi_\mu^{-1*}\left(\frac{(\kappa_\iota\kappa_{\iota^\prime}|a|^2)(q^*(t, x, \eta), y)}{|D_\Phi (t, q^*(t, x, \eta), y)|}\right). 
\]

Since the set 
\[
K_{\iota\iota^\prime}:=\{\, y\in M \mid \exists x\in M \text{ s.t. } (x, y)\in \rmop{supp} \kappa_\iota \cup \rmop{supp} \kappa_{\iota^\prime} \,\} \subset V_\mu
\]
is compact, we can take $\chi_{\iota\iota^\prime}\in C_c^\infty (V_\mu)$ such that $\chi_{\iota\iota^\prime}=1$ on $K_{\iota\iota^\prime}$. The support of the integral kernel 
\[
(y, z)\longmapsto \frac{1}{(2\pi \hbar t)^{n/2}}\int_M \overline{a(t, x, y)}b(t, x, z)e^{i(\Phi (t, x, z)-\Phi(t, x, y))/\hbar t}\, \vol_g(x)
\]
of the operator $T_\hbar [\kappa_\iota a](t)^*T_\hbar [\kappa_{\iota^\prime}a](t)$ is included in the set
\[B:=\{\, (y, z)\in M\times M \mid \exists x\in M \text{ s.t. } (x, y)\in \rmop{supp} \kappa_\iota , \, (x, z)\in \rmop{supp} \kappa_{\iota^\prime}\,\}
\]
Since $B\subset K_{\iota\iota^\prime}\times K_{\iota\iota^\prime}$, we have
\begin{align*}
T_\hbar [\kappa_\iota a](t)^* T_\hbar [\kappa_{\iota^\prime} a](t)
&=\chi_{\iota\iota^\prime}T_\hbar [\kappa_\iota a](t)^* T_\hbar [\kappa_{\iota^\prime} a](t)\chi_{\iota\iota^\prime} \\
&=\chi_{\iota\iota^\prime}\phi_\mu^*b_{\iota\iota^\prime}^\mathrm{w}(\hbar t; t, y, \hbar t D_y)\phi_\mu^{-1*}\chi_{\iota\iota^\prime}
\end{align*}
recalling \eqref{eq_weyl_psiDO_ht}. 
\end{proof}

\subsection{Proof of Theorem \ref{theo_psiDO}, Theorem \ref{theo_L2_boundedness} and Theorem \ref{theo_vanishing}}\label{subsec_proof_FIO}

Now we are ready to the proof of the estimates of oscillatory integral operators in Subsection \ref{subsec_osc_goal}. 
In the following we fix $\{\Omega_\iota\}_{\iota\in I}$ as in Theorem \ref{theo_nonempty_intersection}. For each $(\iota, \iota^\prime)\in I\times I$ such that $\Omega_\iota \cap \Omega_{\iota^\prime}\neq\varnothing$, we fix $(\lambda, \mu )\in \Lambda\times \Lambda$ in the statement 2 of Theorem \ref{theo_nonempty_intersection} and denote by $(\lambda_{\iota\iota^\prime}, \mu_{\iota\iota^\prime})$. 

\begin{proof}[Proof of Theorem \ref{theo_psiDO}]
By Theorem \ref{theo_nonempty_intersection}, we can represent $T_\hbar[a](t)^*T_\hbar[a](t)$ as
\begin{equation}\label{eq_decomposition_omega}
    \begin{split}
&T_\hbar[a](t)^*T_\hbar[a](t) \\
&=\sum_{\Omega_\iota\cap \Omega_{\iota^\prime}\neq \varnothing} \chi_{\iota\iota^\prime}\phi_{\iota\iota^\prime}^*b_{\iota\iota^\prime}^\mathrm{w}(\hbar t; t, y, \hbar t D_y)\phi_{\iota\iota^\prime}^{-1*}\chi_{\iota\iota^\prime}
+\sum_{\Omega_\iota\cap \Omega_{\iota^\prime}= \varnothing} T_\hbar [\kappa_\iota a](t)^*T_\hbar [\kappa_{\iota^\prime}a](t). 
    \end{split}
\end{equation}
Here we put $\phi_{\iota\iota^\prime}:=\phi_{\mu_{\iota\iota^\prime}}$. 
The second term is in $(\hbar t)^\infty \Psi_{\hbar t}^{-\infty}$ since $T_\hbar [\kappa_\iota a](t)^*T_\hbar [\kappa_{\iota^\prime}a](t)$ with $\Omega_\iota \cap \Omega_{\iota^\prime}=\varnothing$ is a composition of oscillatory integral operators with amplitudes with disjoint supports and thus we can apply Theorem \ref{theo_disjoint_support}. Hence by the definition of semiclassical pseudodifferential operators (\ref{defi_psiDO}), $T_\hbar [a](t)^*T_\hbar [a](t)$ is an $(\hbar t)$-pseudodifferential operator of class $\Psi_{\hbar t}^0$. The principal symbol is
\[
\sum_{\Omega_\iota\cap \Omega_{\iota^\prime}\neq \varnothing} \chi_{\iota\iota^\prime}(y)^2 \frac{(\kappa_\iota\kappa_{\iota^\prime}|a|^2)(q^*(t, x, \eta), y)}{|D_\Phi (t, q^*(t, x, \eta), y)|}
=\frac{|a(q^*(t, x, \eta), y)|^2}{|D_\Phi (t, q^*(t, x, \eta), y)|}. 
\]
by the definition \eqref{eq_defi_principal_symbol} of the principal symbol and the formula \eqref{eq_principal_symbol_locally}. 
\end{proof}

\begin{proof}[Proof of Theorem \ref{theo_L2_boundedness}]
We define $A_\hbar (t):=T_\hbar [a](t)^*T_\hbar [a](t)$. Then 
\begin{equation}\label{eq_osc_psiDO}
    \|T_\hbar [a](t)\|_{L^2(M)\to L^2(M)}^2=\|A_\hbar (t)\|_{L^2(M)\to L^2(M)}. 
\end{equation} 
Now we carry out the argument in Proposition E.24 in \cite{Dyatlov-Zworski19}. Since $a\in C_c^\infty([0, t_0]\times N)$, the Schwartz kernel of $A_\hbar (t)$ is compactly supported. Take the supremum of the principal symbol of $A_\hbar (t)$: 
\[
M(t):=\left\|\frac{|a(t, q^*(t, x, \eta), y)|^2}{|D_\Phi (t, q^*(t, x, \eta), y)|}\right\|_{L^\infty(T^*M)}. 
\]
Then the principal symbol of the $(\hbar t)$-pseudodifferential operator $M(t)^2-A_\hbar (t)^* A_\hbar (t)$ is calculated as 
\[
\sigma_{\hbar t} ( M(t)^2-A_\hbar (t)^* A_\hbar (t) )=M(t)^2-\left| \frac{|a(q^*(t, x, \eta), y)|^2}{|D_\Phi (t, q^*(t, x, \eta), y)|}\right|^2\geq 0
\]
for all $(t, y, \eta)$. Since $M(t)^2-A_\hbar (t)^* A_\hbar (t)$ has a compact support (recall that we assumed the compactness of $M$), we can apply the sharp G\aa rding inequality (Theorem \ref{theo_psiDO_basic}) and obtain
\[
\jbracket{(M(t)^2-A_\hbar (t)^* A_\hbar (t) )u, u}_{L^2(M)}\geq -C\hbar t\|u\|_{L^2(M)}^2. 
\]
Hence 
\[
\|A_\hbar (t)u\|_{L^2(M)}^2\leq M(t)^2\| u\|_{L^2(M)}^2+C\hbar t\|u\|_{L^2(M)}^2
\leq (M(t)^2+C\hbar t)\|u\|_{L^2(M)}^2. 
\]
Since this holds for all $u\in L^2 (M)$, we have
\[
\|A_\hbar (t)\|_{L^2(M)\to L^2(M)}\leq \sqrt{M(t)^2+C\hbar t}. 
\]
Thus by \eqref{eq_osc_psiDO}, we obtain
\[
\|T_\hbar [a](t)\|_{L^2(M)\to L^2(M)}^4=\|A_\hbar (t)\|_{L^2(M)\to L^2(M)}^2 
\leq M(t)^2+C\hbar t. \qedhere
\]
\end{proof}

\begin{proof}[Proof of Theorem \ref{theo_vanishing}]
We divide the proof into 7 steps. 

\textbf{Step 1.}    As in the proof of Theorem \ref{theo_psiDO}, we employ the decomposition \eqref{eq_decomposition_omega}. 
Then 
\begin{align}
    &\| T_\hbar [a](t)\|_{H_\hbar^{s/2}\to L^2}^2=\|T_\hbar [a](t)(1-\hbar^2\triangle_g)^{-s/4}\|_{L^2\to L^2}^2 \nonumber\\
    &=\| (1-\hbar^2\triangle_g)^{-s/4}T_\hbar [a](t)^* T_\hbar [a](t) (1-\hbar^2\triangle_g)^{-s/4}\|_{L^2\to L^2} \nonumber\\
    &\leq C\sum_{\Omega_\iota \cap \Omega_{\iota^\prime}\neq \varnothing} \|b_{\iota\iota^\prime}^\mathrm{w}(\hbar t; t, y, \hbar t D_y)\|_{H_\hbar^{s/2}\to H_\hbar^{-s/2}} \nonumber\\
    &\quad+ \sum_{\Omega_\iota \cap \Omega_{\iota^\prime}=\varnothing}  \| (1-\hbar^2\triangle_g)^{-s/4}T_\hbar [a_\iota](t)^* T_\hbar [a_{\iota^\prime}](t) (1-\hbar^2\triangle_g)^{-s/4}\|_{L^2\to L^2} \nonumber\\
    \begin{split}&\leq C\sum_{\Omega_\iota \cap \Omega_{\iota^\prime}\neq \varnothing} \|b_{\iota\iota^\prime}^\mathrm{w}(\hbar t; t, y, \hbar t D_y)\|_{H_\hbar^{s/2}\to H_\hbar^{-s/2}} \\
    &\quad+ \sum_{\Omega_\iota \cap \Omega_{\iota^\prime}=\varnothing}  \| T_\hbar [a_\iota](t)^* T_\hbar [a_{\iota^\prime}](t) \|_{L^2\to L^2}\end{split}\label{eq_decomposition_beginning}
    \end{align}
    for some constant $C>0$ independent of $t\in [0, t_0/2]$ and $\hbar \in (0, 1]$. 

    \textbf{Step 2.} We first consider the case $\Omega_\iota \cap \Omega_{\iota^\prime}=\varnothing$. In this case, $T_\hbar[\kappa_\iota a](t)^* T_\hbar[\kappa_{\iota^\prime} a](t)$ is in the class $(\hbar t)^\infty \Psi_{\hbar t}^{-\infty}$ by Theorem \ref{theo_disjoint_support}. Since $M$ is compact, this leads to the boundedness of the operator: 
    \begin{align}
    \| T_\hbar[a_\iota](t)^* T_\hbar[a_{\iota^\prime}](t)\|_{L^2(M)\to L^2(M)} 
    &\leq \| K_{\iota\iota^\prime}(\hbar t; t, \cdot, \cdot)\|_{L^2(M\times M)} \nonumber\\ 
    &\leq \rmop{Vol}(M) \| K_{\iota\iota^\prime}(\hbar t; t, \cdot, \cdot)\|_{L^\infty (M\times M)} \nonumber\\
    &\leq C(\hbar t)^{J+1}. \label{eq_sobolev_estimate_empty} 
    \end{align}
    Here $K_{\iota\iota^\prime}(\hbar t; t, \cdot, \cdot)$ is the integral kernel of $ T_\hbar[\kappa_\iota a](t)^* T_\hbar[\kappa_{\iota^\prime} a](t)$ and $\rmop{Vol}(M)$ is the volume of $M$ with respect to the metric $g$. 

\textbf{Step 3.} Next we consider the case $\Omega_\iota \cap \Omega_{\iota^\prime}\neq \varnothing$. Assume that $\partial_t^l\partial_x^\alpha \partial_y^\beta a (0, x, y)|_{x=y}=0$ for all $x\in M$ and $l+|\alpha|+|\beta|\leq 2J+L$. This assumption implies $\partial_t^l \partial_\eta^\alpha b_{\iota\iota^\prime, j} (0, y, 0)=0$ for all $l+|\alpha|\leq L$ and $j\leq J$ by the statement 2 of Theorem \ref{theo_nonempty_intersection}. Then, by the Taylor theorem, we have  
\begin{align*}
&b_{\iota\iota^\prime, j}(t, y, \eta)= \\
&\frac{1}{L!}\sum_{|\alpha|+l=L+1} 
t^l\eta^\alpha
\begin{pmatrix}
L+1 \\
l
\end{pmatrix}
\int_0^1 (1-s)^L\partial_t^l\partial_\eta^\alpha b_{\iota\iota^\prime, j}(st, y, s\eta)\,ds
\end{align*}
for all $j\leq J$. We define
\[
c_{\iota\iota^\prime, jl\alpha}(t, y, \eta):= 
\frac{1}{L!}
\begin{pmatrix}
L+1 \\
l
\end{pmatrix}
\int_0^1 (1-s)^L\partial_t^l \partial_\eta^\alpha b_{\iota\iota^\prime, j}(st, y, s\eta)\,ds. 
\]
Then 
\begin{equation}\label{eq_defi_bc}
b_{\iota\iota^\prime, j}(t, y, \eta)=\sum_{|\alpha|+l=L+1} t^l \eta^\alpha c_{\iota\iota^\prime, jl\alpha}(t, y, \eta)
\end{equation}
holds for all $(t, y, \eta)\in [-t_0/2, t_0/2]\times \mathbb{R}^{2n}$. We can take a large $R>0$ such that 
\begin{equation}\bigcup_{|t|\leq t_0/2} \rmop{supp} b_{\iota\iota^\prime, j}(t, \cdot, \cdot)\subset \{ |\eta|\leq R\}. \label{eq_support_low_energy}\end{equation}

The existence of such $R>0$ is proved as follows. The formula \eqref{eq_expansion_symbol_term} implies that 
\begin{equation} \label{eq_support_1}
    \rmop{supp} b_{\iota\iota^\prime, j}(t, \cdot, \cdot)\subset \{\, (y, \eta)\in \mathbb{R}^{2n}\mid (t, y, \eta, y)\in \rmop{supp}\tilde b_{\iota\iota^\prime}\,\}. \end{equation}
Next we recall the formula \eqref{eq_bisymbol} and the definition \eqref{eq_changing_variables} of $\hat x(t, y, \eta, z)$ and obtain 
\begin{equation}\label{eq_support_2}
    \begin{split}
    &\{\, (y, \eta)\in \mathbb{R}^{2n}\mid (t, y, \eta, y)\in \rmop{supp}\tilde b_{\iota\iota^\prime}\,\} \\
    &\subset \{\, (y, \eta)\in \mathbb{R}^{2n} \mid (t, \hat x(t, y, \eta, y), y) \in \rmop{supp}(a_\iota a_{\iota^\prime})\,\}.
    \end{split}
\end{equation}
for all $[-t_0/2, t_0/2]$. 
\eqref{eq_support_1} and \eqref{eq_support_2} imply that 
\begin{equation}\label{eq_support_12}
    \rmop{supp}b_{\iota\iota^\prime, j}(t, \cdot, \cdot) \subset \{\, (y, \eta)\in \mathbb{R}^{2n} \mid (t, \hat x(t, y, \eta, y), y) \in \rmop{supp}(a_\iota a_{\iota^\prime})\,\}\end{equation}
for all $t\in [-t_0/2, t_0/2]$. Let $x=\hat x(t, y, \eta, y)$. Then by the defition \eqref{eq_changing_variables} of $\hat x(t, y, \eta, y)$ implies $\eta=-\partial_y \Phi_{\lambda\mu}(t, x, y)$. Thus if $(t, \hat x(t, y, \eta, y), y) \in \rmop{supp}(a_\iota a_{\iota^\prime})$, then 
\begin{equation}\label{eq_support_3}
\eta\in \{ \, -\partial_y \Phi_{\lambda\mu} (t, x, y) \mid (t, x, y)\in \rmop{supp}(a_\iota a_{\iota^\prime})\,\}. 
\end{equation}
The right hand side of \eqref{eq_support_3} is an image of the compact set $\rmop{supp}(a_\iota a_{\iota^\prime})$ by the continuous map $-\partial_y\Phi_{\lambda\mu}(t, x, y)$. Thus it is a compact set. Hence we can take $R>0$ such that 
\begin{equation}\label{eq_support_123}
    \begin{split}
        &\{ \, (t, y, \eta)\in [-t_0/2, t_0/2]\times \mathbb{R}^{2n} \mid (t, \hat x(t, y, \eta, y), y) \in \rmop{supp} (a_\iota a_{\iota^\prime})\,\} \\
&\subset \{ \, (t, y, \eta)\in [-t_0/2, t_0/2]\times \mathbb{R}^{2n} \mid |\eta|\leq R\,\}. 
    \end{split} \end{equation}
Combining \eqref{eq_support_12} and \eqref{eq_support_123}, we obtian 
\[\rmop{supp}b_{\iota\iota^\prime, j}(t, \cdot, \cdot) \subset \{ \, (y, \eta)\in \mathbb{R}^{2n} \mid |\eta|\leq R\,\}
    \]
for all $t\in [-t_0/2, t_0/2]$. Thus \eqref{eq_support_low_energy} is proved. 

\textbf{Step 4.} In the following steps, we should forget that $t$ sometimes behaves like ``\textit{a semiclassical parameter}." We consider the semiclassical Weyl quantization of the symbol
\[
(t\eta)^\alpha c_{\iota\iota^\prime, jl\alpha}(t, y, t\eta). 
\]
We show that for all mutiindices $\alpha$, $\beta$ and $\gamma$, there exists $C_{\alpha\beta\gamma}>0$ such that 
\begin{equation}\label{eq_estimate_c}
|\partial_y^\beta \partial_\eta^\gamma ((t\eta)^\alpha c_{\iota\iota^\prime, j\alpha}(t, y, t\eta))|
\leq C_{\alpha\beta\gamma} t^{s_\alpha} (1+|\eta|)^s 
\end{equation}
holds for all $y\in \mathbb{R}^n$, $t\in (0, t_0]$ and $|\eta|\leq R/t$. 
For any multiindices $\alpha$, $\beta$ and $\gamma$, and $y\in \mathbb{R}^n$ and $|\eta|\leq R/t$, we have  
\begin{align}
&|\partial_y^\beta \partial_\eta^\gamma ((t\eta)^\alpha c_{\iota\iota^\prime, jl\alpha}(t, y, t\eta))|
= |t^{|\gamma|}\partial_y^\beta \partial_{\tilde\eta}^\gamma ({\tilde\eta}^\alpha c_{\iota\iota^\prime, j\alpha}(t, y, \tilde\eta))|_{\tilde\eta=t\eta}| \nonumber\\
&\leq C_{\alpha\beta\gamma}t^{|\gamma|}\sum_{\gamma^\prime \leq \gamma} 
| t^{|\alpha|-|\gamma|+|\gamma^\prime|}(\partial_\eta^{\gamma-\gamma^\prime} \eta^\alpha) (\partial_y^\beta \partial_\eta^{\gamma^\prime} c_{\iota\iota^\prime, j\alpha}(t, y, t\eta)) | \nonumber\\
&\leq C_{\alpha\beta\gamma}\sum_{\substack{0\leq \gamma^\prime \leq \gamma \\ \gamma^\prime \geq \gamma-\alpha}} 
t^{|\alpha|+|\gamma^\prime|}|\eta|^{|\alpha|-|\gamma|+|\gamma^\prime|} |c_{\iota\iota^\prime, j\alpha}|_{\beta, \gamma^\prime}. \label{eq_estimate_c_wip}
\end{align}
Here 
\[
|c_{\iota\iota^\prime, j\alpha}|_{\beta, \gamma^\prime}
:=\sup_{\substack{t \in [-t_0, t_0] \\ (y, \eta)\in \mathbb{R}^{2n} \\ |\eta|\leq R}}| \partial_y^\beta \partial_\eta^{\gamma^\prime} c_{\iota\iota^\prime, j\alpha}(t, y, \eta)|<\infty. 
\]
If $|\alpha|-|\gamma|+|\gamma^\prime|\geq s \, (\geq s_\alpha)$, then for $|t\eta|\leq R$, 
\begin{align}
t^{|\alpha|+|\gamma^\prime|}|\eta|^{|\alpha|-|\gamma|+|\gamma^\prime|} 
&= t^{|\gamma|}|t\eta|^{|\alpha|-|\gamma|+|\gamma^\prime|} 
\leq t_0^{|\gamma|-s_\alpha+s}R^{|\alpha|-|\gamma|+|\gamma^\prime|-s_\alpha}t^{s_\alpha} |\eta|^s \nonumber\\
&\leq C_{\alpha\beta\gamma}t^{s_\alpha}(1+|\eta|)^s \label{eq_t_eta_case1}
\end{align}
holds. 

If $(0\leq )\, |\alpha|-|\gamma|+|\gamma^\prime|< s$, then, recalling $s_\alpha\leq |\alpha|$, we obtain 
\begin{align}
t^{|\alpha|+|\gamma^\prime|}|\eta|^{|\alpha|-|\gamma|+|\gamma^\prime|} 
&\leq t_0^{|\gamma^\prime|+|\alpha|-s_\alpha} t^{s_\alpha}(1+|\eta|)^{s_\alpha}
\leq t_0^{|\gamma^\prime|+|\alpha|-s_\alpha} t^{s_\alpha}(1+|\eta|)^s \nonumber\\
&\leq C_{\alpha\beta\gamma}t^{s_\alpha} (1+|\eta|)^s. \label{eq_t_eta_case2}
\end{align}
for all $|t\eta|\leq R$. 
Thus, by \eqref{eq_t_eta_case1} and \eqref{eq_t_eta_case2}, $|\alpha|-|\gamma|+|\gamma^\prime|\geq 0$ implies that there exists a constant $C_{\alpha\beta\gamma}>0$ such that the inequality
\begin{equation}\label{eq_t_eta_estimate}
    t^{|\alpha|+|\gamma^\prime|}|\eta|^{|\alpha|-|\gamma|+|\gamma^\prime|}
    \leq C_{\alpha\beta\gamma} t^{s_\alpha}(1+|\eta|)^s
\end{equation}
holds for all $t\in (0, t_0/2]$ and $|\eta|\leq R/t$. Therefore the inequalities \eqref{eq_estimate_c_wip} and \eqref{eq_t_eta_estimate} imply
\begin{align*}
    &|\partial_y^\beta \partial_\eta^\gamma ((t\eta)^\alpha c_{\iota\iota^\prime, jl\alpha}(t, y, t\eta))| \\
    &\leq C_{\alpha\beta\gamma} \sum_{\substack{0\leq \gamma^\prime \leq \gamma \\ \gamma^\prime \geq \gamma-\alpha}} 
    t^{|\alpha|+|\gamma^\prime|}|\eta|^{|\alpha|-|\gamma|+|\gamma^\prime|} |c_{\iota\iota^\prime, j\alpha}(t, \cdot, \cdot)|_{\beta, \gamma^\prime} \\
    &\leq C_{\alpha\beta\gamma} t^{s_\alpha}(1+|\eta|)^s. 
\end{align*}
The inequality \eqref{eq_estimate_c} is proved. 

\textbf{Step 5.} By \eqref{eq_defi_bc} ($\eta$ is replaced by $t\eta$) and \eqref{eq_estimate_c}, the derivatives of $b_{\iota\iota^\prime, j}(t, y, t\eta)$ are estimated as
\begin{align}
|\partial_y^\beta \partial_\eta^\gamma (b_{\iota\iota^\prime, j}(t, y, t \eta))|
&\leq \sum_{|\alpha|+l=L+1} t^l  |\partial_y^\beta \partial_\eta^\gamma ((t\eta)^\alpha c_{\iota\iota^\prime, j\alpha}(t, y, t\eta))| \nonumber\\
&\leq C_{\beta\gamma} (1+|\eta|)^s\sum_{|\alpha|+l=L+1} t^{l+s_\alpha} \label{eq_estimate_b}
\end{align}
provided $|t\eta|\leq R$. Since $b_{\iota\iota^\prime, j}(t, y, t\eta)=0$ if $|t\eta|>R$ by \eqref{eq_support_low_energy}, the above inequality \eqref{eq_estimate_b} holds for all $(y, \eta)\in \mathbb{R}^{2n}$. 

The inequality \eqref{eq_estimate_b} means that the family 
\[\left\{ \, \left(\sum_{|\alpha|+l\leq L+1} t^{l+s_\alpha}\right)^{-1}b_{\iota\iota^\prime, j}(t, y, t\eta)\,\right\}_{t\in (0, t_0]}\] 
is bounded in the symbol class 
\[
    S^s_{0, 0}:=\left\{\, a(x, \xi)\in C^\infty(\mathbb{R}^n\times \mathbb{R}^n) \,\middle|\, 
    \begin{aligned}
    &\|(1+|\xi|)^{-m}\partial_x^\alpha\partial_\xi^\beta a(x, \xi)\|_{L^\infty}<\infty \\
    &\text{ for all } \alpha, \beta \in \mathbb{N}_0^n
    \end{aligned}
    \,\right\}. \]

Thus there exists a constant $C>0$ such that 
\begin{equation}\label{eq_sobolev_estimate}
\|b_{\iota\iota^\prime, j}^\mathrm{w}(t, y, \hbar t D_y)\|_{H_\hbar^{s/2}\to H_\hbar^{-s/2}} 
\leq C\sum_{|\alpha|+l=L+1} t^{l+s_\alpha}. 
\end{equation}
for all $t\in (0, t_0/2]$ and $\hbar\in (0, 1]$ (see Theorem 8.10 in \cite{Zworski12}). 

\textbf{Step 6.} The remainder term 
\[
\tilde b_{\iota\iota^\prime, J}(\hbar t; t, y, t\eta):=(\hbar t)^{-J-1}\left(b_{\iota\iota^\prime}(\hbar t; t, y, t\eta)-\sum_{j=0}^J (\hbar t)^j b_{\iota\iota^\prime, j} (t, y, t\eta)\right)
\]
of the asymptotic expansion \eqref{eq_asymptotic_b} is bounded in $S^0_{1, 0}$ by the definition of asymptotic expansions. Thus 
\begin{equation}\label{eq_sobolev_estimate_rem}
\| \tilde b_{\iota\iota^\prime, J}^\mathrm{w}(\hbar t; t, y, \hbar t D_y)\|_{H_\hbar^{s/2}\to H_\hbar^{-s/2}}\leq C
\end{equation}
for some constant $C>0$ independent of $t\in (0, t_0/2]$ and $\hbar\in (0, 1]$. Hence, by combining \eqref{eq_sobolev_estimate} and \eqref{eq_sobolev_estimate_rem}, we obtain
\begin{align}
\|b_{\iota\iota^\prime}^\mathrm{w}(\hbar t; t, y, \hbar t D_y)\|_{H_\hbar^{s/2}\to H_\hbar^{-s/2}} 
&\leq \sum_{j=0}^J (\hbar t)^j \|b_{\iota\iota^\prime, j}^\mathrm{w}(t, y, \hbar t D_y)\|_{H_\hbar^{s/2}\to H_\hbar^{-s/2}} \nonumber\\
&\quad + (\hbar t)^{J+1} \|\tilde b_{\iota\iota^\prime, J}^\mathrm{w}(\hbar t; t, y, \hbar t D_y)\|_{H_\hbar^{s/2}\to H_\hbar^{-s/2}} \nonumber\\
&\leq C\sum_{j=0}^J\sum_{|\alpha|+l=L+1} \hbar^j t^{l+j+s_\alpha}+O((\hbar t)^{J+1}) \nonumber\\
&\leq C\left(\sum_{|\alpha|+l=L+1} t^{l+s_\alpha}+(\hbar t)^{J+1}\right) \label{eq_sobolev_estimate_nonempty}
\end{align}
for some $C<0$ independent of $t\in (0, t_0/2]$ and $\hbar\in (0, 1]$. 

\textbf{Step 7.} Finally we substitude \eqref{eq_sobolev_estimate_empty} and \eqref{eq_sobolev_estimate_nonempty} to \eqref{eq_decomposition_beginning}, we obtain 
\begin{align*}
    \| T_\hbar [a](t)\|_{H_\hbar^{s/2}\to L^2}^2 
    &\leq C\sum_{\Omega_\iota \cap \Omega_{\iota^\prime}\neq \varnothing} \|b_{\iota\iota^\prime}^\mathrm{w}(\hbar t; t, y, \hbar t D_y)\|_{H_\hbar^{s/2}\to H_\hbar^{-s/2}} \\
    &\quad+ \sum_{\Omega_\iota \cap \Omega_{\iota^\prime}=\varnothing}  \| T_\hbar [a_\iota](t)^* T_\hbar [a_{\iota^\prime}](t) \|_{L^2\to L^2} \\
    &\leq C\left(\sum_{|\alpha|+l=L+1} t^{l+s_\alpha}+(\hbar t)^{J+1}\right)+C(\hbar t)^{J+1} \\
    &\leq C\left(\sum_{|\alpha|+l=L+1} t^{l+s_\alpha}+(\hbar t)^{J+1}\right).
\end{align*}
The proof is completed. 
\end{proof}



\section{Definition of short-time approximate solution}\label{sect_def_of_E_h(t)}
\subsection{Precise definition of short-time approximate solution}\label{subsect_definition}

We have proved in Section \ref{sect_low_energy_cm} that the action $S(t, x, y)$ and the Morette-Van Vleck determinant $D(t, x, y)$ are well-defined and have suitable properties. Thus 
we can define a short-time approximate solution to the modified Schr\"odinger equation. 

\begin{defi}
We fix $[-t_0, t_0]\times N\subset \mathbb{R}\times M\times M$ as in Theorem \ref{theo_action_well_defined}. Then for $\chi\in C_c^\infty(N; [0, 1])$ with $\chi=1$ near $\mathrm{diag}(M)$ and $t\in (0, t_0]$, we define
\[
E^\chi_\hbar(t)u(x):=\frac{1}{(2\pi i\hbar)^{n/2}}\int_M \chi(x, y) \sqrt{D(t, x, y)} e^{iS(t, x, y)/\hbar} u(y)\, \vol_g(y). 
\]
If $\chi$ is obvious from the context, we denote $E^\chi_\hbar(t)$ by $E_\hbar(t)$. 
\end{defi}

\subsection{Proof of Proposition \ref{prop_fundamental_properties}}\label{subsect_proof_of_fp}

\begin{proof}[Proof of (i)]
This is an immediate consequence of the Lebesgue dominated convergence theorem. 
\end{proof}

\begin{proof}[Proof of (ii)]
Fix $u\in C^\infty(M)$ and $x\in M$. By Theorem \ref{theo_difference_from_free}, we have
\[
E_\hbar(t)u(x)=\frac{1}{(2\pi i \hbar t)^{n/2}}\int_M \chi(x, y)\sqrt{D_\Phi (t, x, y)}e^{it\phi(t, x, y)/\hbar}e^{id(x, y)^2/ 2\hbar t}u(y) \,\vol_g(y). 
\]
Since we are only interested in the behavior in $t\to +0$, we fix the semiclassical parameter $\hbar>0$. We regard $\chi(x, y)\sqrt{D_\Phi (t, x, y)}e^{it\phi(t, x, y)/\hbar}u(y)$ as an amplitude and $d(x, y)^2/2\hbar$ as a phase function. We apply the method of stationary phase (see \cite{Hormander85-1} for example) as $t\to +0$. Then there exists a constant $C_\hbar>0$ ($C_\hbar$ may diverge as $\hbar\to +0$) such that 
\[
\left\| E_\hbar(t)u(x)- \frac{\sqrt{D_\Phi(t, x, x)}}{|\mathrm{Hess}_{g, y}\Phi(0, x, x)|^{1/2}}e^{it\phi(t, x, x)/\hbar}u(x)\right\|_{L_x^\infty(M)}\leq C_\hbar t
\]
for small $t\ll t_0$. Here we used the fact that $d_y(d(x, y)^2)=0$ if and only if $x=y$ provided $(x, y)\in N$. The weighted Hessian 
\[
\mathrm{Hess}_{g, y}\Phi(0, x, x):=\det (g_{ij}(x))^{-1}\det (\partial_y^2 \Phi(0, x, y))|_{y=x}, 
\]
which is independent of the choice of local coordinates near $x$, is equal to 1. Thus 
\[
\| E_\hbar(t)u-u\|_{L^2(M)}\leq C\|(\sqrt{D_\Phi (t, x, x)}e^{it\phi(t, x, x)/\hbar}-1)u(x)\|_{L^2(M)}+
C_\hbar t. 
\]
Since $\|D_\Phi \|_{L^\infty(I\times N)}<\infty$ and $D_\Phi (0, x, x)=1$ by Proposition \ref{prop_amplitude_diagonal}, we can apply the Lebesgue dominated convergence theorem as $t\to +0$ and obtain
\[
\lim_{t \to +0} \|E_\hbar(t)u-u\|_{L^2(M)}=0. \qedhere
\]
\end{proof}

(iii) is an immediate consequence of the following lemma. 
\begin{lemm}\label{lemm_time_derivative}
Let $u\in C^\infty(M)$. Then $t\in (0, t_0]\mapsto E_\hbar(t)u\in L^2(M)$ is differentiable in $L^2$ and its derived function is the pointwise derivative $v(t, x):=\partial_t(E_\hbar(t)u(x))$. Moreover $t\in (0, t_0]\mapsto v(t, \cdot)\in L^2(M)$ is continuous in $L^2(M)$.  
\end{lemm}

\begin{proof}
$v(t, x)$ is the form
\[
v(t, x)=\int_M b_\hbar(t, x, y) e^{i\Phi(t, x, y)/ \hbar t} u(y) \, \vol_g(y), 
\]
where 
\[ b_\hbar (t, x, y)=\frac{\chi(x, y)}{(2\pi i\hbar)^{n/2}}\left(\frac{\partial}{\partial t}(\sqrt{D(t, x, y)})+\frac{i}{\hbar}\sqrt{D(t, x, y)}\frac{\partial S}{\partial t}(t, x, y)\right). \]
We only have to note that $b_\hbar\in C^\infty((0, t_0]\times N)$. By the Lebesgue dominated convergence theorem, the mapping
\[
t\in (0, t_0) \longmapsto v(t, \cdot)\in L^2(M)
\]
is continuous in $L^2$. Thus we can consider the Riemann integral in the $L^2$: 
\[
w(t, s):=\int_s^t v(\sigma, \cdot)\, d\sigma \in L^2(M)
\]
for $s, t\in (0, t_0]$. We prove that $E(t)u-E(s)u=w(t, s)$. Let $\varphi\in C^\infty(M)$ be an arbitrary test function. Since we can change the order of the integration and the inner product, we have
\[
\jbracket{w(t, s), \varphi}_{L^2(M)}=\jbracket{\int_s^t v(\sigma, \cdot)\, d\sigma, \varphi}_{L^2(M)}
=\int_s^t \jbracket{v(\sigma, \cdot), \varphi}_{L^2(M)}\, d\sigma. 
\]
By Fubini's theorem, we have
\begin{align*}
\int_s^t \jbracket{v(\sigma, \cdot), \varphi}_{L^2(M)}\, d\sigma
&=\int_M \vol_g(x)\, \overline{\varphi(x)} \int_s^t d\sigma \,v(\sigma, x) \\
&=\int_M (E_\hbar(t)u(x)-E_\hbar(s)u(x))\overline{\varphi(x)} \, \vol_g(x) \\
&=\jbracket{E_\hbar (t)u-E_\hbar(s)u, \varphi}_{L^2(M)}. 
\end{align*}
Since $\varphi\in C^\infty(M)$ is arbitrary, we obtain $w(t, s)=E_\hbar(t)u-E_\hbar(s)u$. 

This fact means that $t\in (0, t_0)\mapsto E_\hbar(t)u\in L^2(M)$ is differentiable in the $L^2$ norm topology and the derived function
\[
t \in (0, t_0) \longmapsto \frac{\partial}{\partial t}E_\hbar(t)u=v(t, \cdot)\in L^2(M)
\]
is continuous in $L^2(M)$. 
\end{proof}


\section{Proof of stability and consistency}\label{sect_proof_s_and_c}

We prove Theorem \ref{theo_stability_consistency} in this section. We take $\delta:=t_0/2>0$. We prove that this $\delta$ is a required one in the statements in Theorem \ref{theo_stability_consistency}.  

\subsection{Proof of stability}

\begin{proof}
We apply Theorem \ref{theo_L2_boundedness} to $E_\hbar(t)$, setting $a=\chi\sqrt{D_\Phi}$. Then we have \begin{align*}
\| E_\hbar(t)\|_{L^2(M)\to L^2(M)}^4 
&\leq \left\|\frac{|\chi (q^t_1(y, \eta), y)|^2|D_\Phi (t, q^t_1(y, \eta), y)|}{|D_\Phi (t, q^t_1(y, \eta), y)|}\right\|_{L^\infty(T^*M)}^2+C\hbar t \\
&=\|\chi (q^t_1(y, \eta), y)^2\|_{L^\infty(T^*M)}^2+C\hbar t \leq 1+C\hbar t\leq e^{C\hbar t}. 
\end{align*}
Hence we obtain
\[
\|E_\hbar (t)\|_{L^2(M)\to L^2(M)}\leq e^{C\hbar t}
\]
for all $t\in (0, t_0/2]$ and $\hbar \in (0, 1]$. 
\end{proof}

\subsection{Proof of consistency}

First we confirm that the integral kernel of $E_\hbar (t)$ satisfies the modified Schr\"odinger equation approximately. 
\begin{prop}\label{prop_WKB}
    We have 
\[
\left( i\hbar\frac{\partial}{\partial t}-\tilde H_\hbar\right)(\sqrt{D}\chi e^{iS/\hbar})
= 
(\hbar^2 t^{-n/2} r_0(t, x, y)+ \hbar t^{-n/2-1} r_{1, \hbar}(t, x, y))e^{iS/\hbar},  
\]
where
\begin{equation}\label{eq_defi_r0}
r_0(t, x, y)=\chi(x, y)\left(\frac{1}{2}\triangle_x a(t, x, y)-\frac{1}{12}R(x)a(t, x, y)\right)
\end{equation}
and
\[
r_{1, \hbar}(t, x, y)=g(ia\,\mathrm{grad}_x\Phi+\hbar t\,\mathrm{grad}_xa, \mathrm{grad}_x\chi)+\frac{\hbar t}{2}a\triangle_x\chi. 
\]
\end{prop}

\begin{proof}
By the Leibnitz rule, we have
\begin{align*}
&e^{-iS/\hbar}\left( i\hbar\frac{\partial}{\partial t}-\tilde H_\hbar\right)(\sqrt{D}\chi e^{iS/\hbar})= \\
& -\left(\frac{\partial S}{\partial t}+\frac{1}{2}|\mathrm{grad}_x S|_g^2+V\right) \sqrt{D}\chi \\
&+i \hbar\left(\frac{\partial \sqrt{D}}{\partial t}+g(\mathrm{grad}_xS, \mathrm{grad}_x \sqrt{D})+\frac{1}{2}\sqrt{D}\triangle_x S\right) \chi \\ 
&+\hbar^2t^{-n/2} r_0(t, x, y)+ \hbar t^{-n/2-1} r_{1, \hbar}(t, x, y). 
\end{align*}
Since $S$ solves the Hamilton-Jacobi equation 
\[
\frac{\partial S}{\partial t}(t, x, y)+\frac{1}{2}|\mathrm{grad}_x S(t, x, y)|_g^2+V(x)=0
\]
by Theorem \ref{theo_property_of_action} and $D^{1/2}$ satisfies the transport equation
\begin{align*}
&\frac{\partial \sqrt{D}}{\partial t}(t, x, y)+g\left(\mathrm{grad}_xS(t, x, y), \mathrm{grad}_x \sqrt{D(t, x, y)}\right) \\
&+\frac{1}{2}\sqrt{D(t, x, y)}\triangle_x S(t, x, y) 
=0
\end{align*}
by Theorem \ref{theo_1/2_transport}, we obtain the conclusion. 
\end{proof}

Recall the definition \eqref{eq_defi_oscillatory} of $T_\hbar [a](t)$ for an amplitude function $a(t, x, y)$:
\[
T_\hbar [a](t)u(x)=\frac{1}{(2\pi i\hbar t)^{n/2}}\int_M a(t, x, y)e^{i\Phi(t, x, y)/\hbar t}u(y)\, \vol_g(y). 
\]
Then Proposition \ref{prop_WKB} implies  
\[
\left( i\hbar\frac{\partial}{\partial t}-\tilde H_\hbar\right)(E_\hbar(t)u(x))
=\hbar^2 T_\hbar[r_0](t)u(x)+\hbar t^{-1} T_\hbar[r_{1, \hbar}](t)u(x)
\]
for all $u\in C^\infty(M)$ and $x\in M$. We define an operator $G_\hbar (t)$ as 
\[
G_\hbar(t)u(x):=\hbar^2 T_\hbar [r_0](t)u(x)+ \hbar t^{-1} T_\hbar [r_1](t)u(x)
\]
for $u\in C^\infty(M)$. Then 
\begin{equation}\label{eq_remainder_schroedinger}
\|G_\hbar(t)u\|_{L^2(M)}\leq \hbar^2\|T_\hbar [r_0](t)u\|_{L^2(M)}+\hbar t^{-1}\| T_\hbar [r_1](t)u\|_{L^2(M)}. 
\end{equation}

In the following we estimate $\|T_\hbar [r_0](t)u\|_{L^2(M)}$ and $\| T_\hbar [r_1](t)u\|_{L^2(M)}$ employing Theorem \ref{theo_vanishing}. 

\begin{lemm}\label{lemm_r1}
For any $\varepsilon\in [0, 1/2]$, there exists a constant $C>0$ such that for all $t\in (0, t_0/2]$ and $\hbar\in (0, 1]$, the inequality
\[
\| T_\hbar[r_1](t)\|_{H_\hbar^{1+\varepsilon}(M)\to L^2(M)}\leq Ct^{1+\varepsilon}
\]
holds. 
\end{lemm}

\begin{proof}
Noting that $r_1$ vanishes near $\{0\}\times \rmop{diag}(M)$, we apply Theorem \ref{theo_vanishing} setting $J=1$, $L=2$ and
\[
s_\alpha=
\begin{cases}
0 & \text{for } |\alpha|=0, \\
|\alpha|-1+2\varepsilon & \text{for } 1\leq |\alpha|\leq 3. \\
\end{cases}
\]
Then there exists a constant $C>0$ such that 
\begin{align*}
\|T[r_1](t)u\|_{L^2(M)}^2 
&\leq 
C(t^3+t^2\times (\hbar t)^{2\varepsilon}+t\times t^{1+2\varepsilon}+t^{2+2\varepsilon}+(\hbar t)^2)\|u\|_{H^{1+\varepsilon}(M)}^2 \\
&\leq Ct^{2+2\varepsilon}\|u\|_{H^{1+\varepsilon}(M)}^2
\end{align*}
holds for all $u\in C^\infty(M)$, $t\in (0, t_0/2]$ and $\hbar\in (0, 1]$. 
\end{proof}

Next we estimate $T_\hbar [r_0](t)$. 

\begin{lemm}\label{lemm_r0}
For any $\varepsilon\in [0, 1/2]$, there exists a constant $C>0$ such that 
\[
\|T_\hbar [r_0](t)\|_{H_\hbar^\varepsilon(M)\to L^2(M)} \leq C t^\varepsilon
\]
for all $t\in (0, t_0/2]$ and $\hbar\in (0, 1]$. 
\end{lemm}

\begin{proof}
First, since $r_0(t, x, y)=r_0(0, x, y)+O(t)$ by the Taylor theorem, we have
\[
\|T_\hbar[r_0](t)\|_{H_\hbar^\varepsilon\to L^2}\leq \|T_\hbar[r_0(0, \cdot, \cdot)](t)\|_{H_\hbar^\varepsilon\to L^2}+O(t). 
\]
By Theorem \ref{theo_curvature} and the definition \eqref{eq_defi_r0} of $r_0$, $r_0(0, x, x)=0$ for all $x\in M$. Thus we can apply Theorem \ref{theo_vanishing} setting $J=L=0$ and 
\[
s_\alpha=
\begin{cases}
0 & \text{for } |\alpha|=0, \\
2\varepsilon & \text{for } |\alpha|=1. 
\end{cases}
\]
Then there exists a constant $C>0$ such that 
\[
\|T_\hbar[r_0(0, \cdot, \cdot)](t)\|_{H_\hbar^\varepsilon\to L^2}^2\leq C(t+t^{2\varepsilon}+\hbar t)\leq Ct^{2\varepsilon}
\]
for all $t\in (0, t_0/2]$ and $\hbar \in (0, 1]$. 
\end{proof}

Now we are ready to prove the consistency stated in Theorem \ref{theo_stability_consistency}. 

\begin{proof}[Proof of Consistency]
Let $u\in C^\infty(M)$. By Lemma \ref{lemm_time_derivative}, the time derivative of $E_\hbar(t)u$ in $L^2(M)$ is $(i\hbar)^{-1}(\tilde H_\hbar E_\hbar(t)u+G_\hbar(t)u)$. Thus we obtain
\[
    G_\hbar (t)u=i\hbar \frac{\partial}{\partial t}E_\hbar (t)u-\tilde H_\hbar u 
\]
as an equation in $L^2(M)$. We employ Lemma \ref{lemm_r1} and Lemma \ref{lemm_r0} to the inequality \eqref{eq_remainder_schroedinger} and obtain 
\[\| G_\hbar(t)u\|_{L^2(M)}\leq C(\hbar^2\times t^\varepsilon+\hbar t^{-1}\times t^{1+\varepsilon})\|u\|_{H_\hbar^{1+\varepsilon}(M)}\leq C\hbar t^{1+\varepsilon}\|u\|_{H_\hbar^{1+\varepsilon}(M)}\]
for some constant $C>0$ independent of $t\in (0, t_0/2]$, $\hbar \in (0, 1]$ and $u\in C^\infty (M)$. 
\end{proof}





\subsection*{Acknowledgments}
The author thanks Professor Kenichi Ito, Professor Shu Nakamura and Professor Yoshihisa Miyanishi for a lot of valuable discussions, advices and comments. He also thanks to Professor Masaki Kawamoto for his encouragements. 




\end{document}